\documentclass{scrartcl}

\KOMAoptions{paper=letter}

\usepackage{mathtools}
\usepackage{thm-restate}
\usepackage{eurosym}

\usepackage{fix-cm}
\usepackage[utf8]{inputenc}
\usepackage[T1]{fontenc}
\usepackage{microtype} 
\usepackage{xspace}
\usepackage{amssymb}
\usepackage{amsmath}
\usepackage{amsthm}
\usepackage[hidelinks]{hyperref} % must be included before cleveref
\usepackage[round]{natbib}
\usepackage{nicefrac}
\usepackage{tikz}
\usetikzlibrary{positioning,arrows,shapes,decorations,calc,fit,arrows.meta,colorbrewer} 
\usepackage{colortbl}
\usepackage{booktabs}
\usepackage{cleveref}
\usepackage{array}
\usepackage{graphicx}
\usepackage{xcolor}
\usepackage{enumitem}
\usepackage{tabularx}

\DeclareMathOperator*{\argmax}{\arg \max}
\DeclareMathOperator*{\argmin}{\arg \min}
\DeclareMathOperator{\sgn}{sgn}
\newcommand{\nash}[1][]{\ifthenelse{\equal{#1}{}}{\ensuremath{\mathit{NASH}}\xspace}{\ensuremath{\mathit{NASH}(#1)}\xspace}}
\newcommand{\util}[1][]{\ifthenelse{\equal{#1}{}}{\ensuremath{\mathit{UTIL}}\xspace}{\ensuremath{\mathit{UTIL}(#1)}\xspace}}
\newcommand{\cut}[1][]{\ifthenelse{\equal{#1}{}}{\ensuremath{\mathit{CUT}}\xspace}{\ensuremath{\mathit{CUT}(#1)}\xspace}}
\newcommand{\anticut}[1][]{\ifthenelse{\equal{#1}{}}{\ensuremath{\mathit{ANTICUT}}\xspace}{\ensuremath{\mathit{ANTICUT}(#1)}\xspace}}
\newcommand{\unc}[1][]{\ifthenelse{\equal{#1}{}}{\ensuremath{\mathit{UNC}}\xspace}{\ensuremath{\mathit{UNC}(#1)}\xspace}}
\newcommand{\ceg}[1][]{\ifthenelse{\equal{#1}{}}{\ensuremath{\mathit{CEG}}\xspace}{\ensuremath{\mathit{CEG}(#1)}\xspace}}
\newcommand{\egal}[1][]{\ifthenelse{\equal{#1}{}}{\ensuremath{\mathit{EGAL}}\xspace}{\ensuremath{\mathit{EGAL}(#1)}\xspace}}
\newcommand{\edr}[1][]{\ifthenelse{\equal{#1}{}}{\ensuremath{\mathit{EDR}}\xspace}{\ensuremath{\mathit{EDR}(#1)}\xspace}}
\newcommand{\Nash}[1][]{\ifthenelse{\equal{#1}{}}{\ensuremath{\mathit{Nash}}\xspace}{\ensuremath{\mathit{Nash}(#1)}\xspace}}
\newcommand{\Egal}[1][]{\ifthenelse{\equal{#1}{}}{\ensuremath{\mathit{Egal}}\xspace}{\ensuremath{\mathit{Egal}(#1)}\xspace}}
\newcommand{\Util}[1][]{\ifthenelse{\equal{#1}{}}{\ensuremath{\mathit{Util}}\xspace}{\ensuremath{\mathit{Util}(#1)}\xspace}}
\newcommand{\gwelfare}{\ensuremath{g\text{-welfare}}\xspace}
\newcommand{\supp}{\mathrm{supp}}

\renewcommand{\epsilon}{\varepsilon}
\newcommand{\deltat}[1][t]{\delta^{#1}}

\usetikzlibrary{arrows.meta}
\tikzset{>={Latex[width=2mm,length=2mm]}}
\usetikzlibrary{decorations.pathmorphing,shapes,snakes}
\usetikzlibrary{calc}
\usetikzlibrary{arrows, decorations.markings}
\usetikzlibrary{positioning}

\theoremstyle{definition}
\newtheorem{definition}{Definition}
\newtheorem{example}{Example}
\theoremstyle{plain}

\newtheorem{lemma}{Lemma}
\newtheorem{proposition}{Proposition}
\newtheorem{corollary}{Corollary}

\newtheorem{observation}{Observation}

\newtheoremstyle{bfnote}% Remark environment
{}{}%
{}{}%
{\bfseries}{.}%
{ }%
{\thmname{#1}\thmnumber{ #2}\thmnote{\textnormal{ (#3)}}}
\theoremstyle{bfnote}
\newtheorem{remark}{Remark}

\usepackage{xurl}
\urldef\polandurl\url{https://pl.wikipedia.org/wiki/Przekazywanie_1%25_podatku_dochodowego_na_rzecz_organizacji_po%C5%BCytku_publicznego_w_Polsce}

\date{} 

\sloppy

\allowdisplaybreaks

\title{\vspace{-4ex}Coordinating Charitable Donations with Leontief Preferences}

\setkomafont{author}{\normalfont\small}

\author{Felix Brandt\\ TUM, Germany \and Matthias Greger\\ Paris-Dauphine, France \and Erel Segal-Halevi\\Ariel Univ., Israel \and Warut Suksompong\\ NUS, Singapore}

\begin{document}

\maketitle

\begin{abstract}

We consider the problem of funding public goods that are complementary in nature. Examples include charities handling different needs (e.g., protecting animals vs.~providing healthcare), charitable donations to different individuals, or municipal units handling different issues (e.g., security vs. transportation).
We model these complementarities by assuming Leontief preferences; that is, each donor seeks to maximize an individually weighted minimum of all contributions across the charities. 
Decentralized funding may be inefficient due to a lack of coordination among the donors; centralized funding may be undesirable as it ignores the preferences of individual donors.
We present a mechanism that combines the advantages of both methods. The mechanism efficiently distributes each donor's contribution so that no subset of donors has an incentive to redistribute their donations. Moreover, it is group-strategyproof, 
satisfies desirable monotonicity properties,
maximizes Nash welfare, returns a unique Lindahl equilibrium, and can be implemented via natural best-response spending dynamics.
\end{abstract}


\section{Introduction}

Private charity, given by individual donors to underprivileged people in their vicinity, has existed long before institutionalized charity via municipal or governmental organizations.
Its main advantage is transparency---the donors know exactly where their money goes to, which may increase their willingness to donate.
A significant disadvantage of private charity is the lack of coordination: donors may donate to certain people or charities without knowing that these recipients have already received ample money from other donors.
Centralized charity via governments or municipalities is potentially more efficient, but if not done carefully, it may disrespect the will of the donors.

\begin{example}
\label{exm:2donors-4charities}
Suppose there are four charities, each of which cares for a distinct set of needy individuals. 
There are two donors: one is willing to contribute $\$900$ and supports charities $A$, $B$, and $C$, whereas the second donor is willing to contribute $\$100$ and supports charities $C$ and $D$. 

\begin{itemize}
\item 
A central organization may collect the donors' contributions and divide them equally among the four charities so that each receives $\$250$. 
While this outcome is the most balanced possible for the charities, it goes against the will of the first donor since $\$150$ of her contribution is used to support charity~$D$.
\item 
By contrast, without coordination, each donor may split her individual contribution equally between the charities she approves. 
As a result, charities~$A$ and $B$ receive $\$300$ each, charity~$C$ receives $\$350$, and charity~$D$ receives $\$50$. 
However, if the second donor knew that charity~$C$ would already receive $\$300$ from the first donor, she would probably prefer to donate more to charity~$D$, for which she is the only contributor.
\item 
Our suggested mechanism would give $\$300$ to each of charities~$A$, $B$, and $C$, and $\$100$ to charity~$D$.
This distribution can be presented as a set of recommendations to the individual donors: the first donor should distribute her contribution uniformly over charities $A$, $B$, and $C$, whereas the second donor should transfer all her contribution to charity~$D$. Importantly, each donor's contribution only goes to charities that the donor approves. Subject to that, the donations are divided as equally as possible. 
\end{itemize}
\end{example}
In the example, we assumed that donors want to \emph{maximize the minimum amount} given to a charity they approve. 
This can be formalized by endowing each donor with a utility function mapping each distribution to the smallest amount of money allocated to one of the donor's approved charities. For example, for the distribution $(300,300,300,100)$, the first agent's utility is $300$ and the second agent's utility is $100$.

Since charitable giving is often driven by egalitarian considerations, minimum-based preferences are reasonable whenever different charities support different individuals. 
Another context where these assumptions are sensible is when charities have different purposes, e.g., nature conservation, education, and healthcare. Each agent prioritizes some of these goals. However, she is not indifferent to how money is distributed among her approved charities but rather sympathizes with the least advantaged ones.

Beyond charitable giving, our model is relevant whenever a group of agents has to fund a set of public goods that are \emph{complementary} in nature. For example, in the context of municipal spending, security, transportation, and public spaces are three complementary public goods, because in order to enjoy the outdoors, all three issues must be adequately handled. Hence, some donors would naturally like to maximize the minimum spending on these three issues.%
\footnote{
As another example, consider a communication network and a set of agents, each of whom intends to transmit a signal along an individual path in the network. Their utilities are given by the quality of the signal at the last node on their path, which equals the minimal transmission quality of an edge along that path. Our mechanism can be used to coordinate agents' investments to improve the transmission quality of edges.
}

Furthermore, our model allows donors to attribute different values than merely $1$ and~$0$ (which indicate approval and disapproval, respectively) to different charities. 
If a donor~$i$ values a charity~$x$ at $v_{i,x}$, then 
$i$'s utility from a distribution $\delta$ equals $\min_x \delta(x)/v_{i,x}$, where the minimum is taken over all charities $x$ for which $v_{i,x}>0$.
Such utilities are known as \emph{Leontief utilities} \citep[see, e.g.,][]{Vari92a,MWG95a} and are often studied in resource allocation problems \citep[e.g.,][]{CoVa04a,Nico04a,GZH+11a,LiXu13a}. 
Whenever $v_{i,x}\in\{0,1\}$ for all agents~$i$ and charities~$x$, we refer to this as (Leontief) utility functions with \emph{binary weights}.%
\footnote{One can further assume that, subject to maximizing the minimum amount given to an approved charity, the donors want to maximize the second-smallest amount, then the third-smallest amount, and so on. Our results carry over to this class of preferences; see \Cref{sec: conclusion}.}

Given each donor's contribution and utility function, our goal is to distribute the money among the charities in a way that respects the individual donors' preferences. For that, we define a game that we call the \emph{distribution game}, in which the strategy space of each donor is the set of all possible distributions of her individual contribution, and the collective outcome is the sum of all individual distributions.
The idea of ``respecting the donors' preferences'' can then be captured by assuring each donor that her individual contribution is distributed optimally, given the individual distributions of the other donors. In other words, the vector of individual distributions is a (pure-strategy) \emph{Nash equilibrium} of the distribution game; the sum of all individual distributions maximizes each donor's utility, given that the other donors follow their own prescriptions.
One can check that in \Cref{exm:2donors-4charities} there is a unique Nash equilibrium, in which the first donor uniformly distributes her contribution among $A$, $B$, and $C$, whereas the second donor allocates her entire contribution to $D$. The sum of these individual distributions is the (unique) equilibrium distribution $(300,300,300,100)$.

Since utility functions are concave, the existence of an equilibrium distribution can be deduced from classic results in equilibrium theory \citep{Debr52a,ArDe54a,Rose65a}. 
Our first main result (\Cref{thm:unique-eq-distribution})  is that each preference profile admits a \emph{unique} equilibrium distribution.%
\footnote{
In general, the Nash equilibrium need not be unique: there may be several vectors of individual distributions, all of which are Nash equilibria of the distribution game. But in all these Nash equilibria, the individual distributions sum up to the same unique equilibrium distribution.
}
Moreover, we prove that the unique equilibrium distribution coincides with the unique distribution that maximizes the product of individual utilities weighted by their contributions (\emph{Nash welfare}), which implies that it is Pareto efficient and can be computed via convex programming. The equilibrium distribution can be viewed as the market equilibrium of a pure public good market, as well as a Lindahl equilibrium with personalized prices.

In \Cref{exm:2donors-4charities}, the equilibrium distribution $(300,300,300,100)$ also maximizes the minimum utility of all agents (\emph{egalitarian welfare}) subject to each donor only contributing to her approved charities. We show that this is true in general when weights are binary. This result extends to an infinite class of welfare measures ``in between'' Nash welfare and egalitarian welfare. 
Moreover, for the case of binary weights, we show that the equilibrium distribution coincides with the distribution that allocates individual contributions to approved charities such that the minimum contribution to charities is maximized lexicographically. This allows for simpler computation via linear programming. 

Based on existence and uniqueness, we can define the \emph{equilibrium distribution rule (\edr)}---the mechanism that returns the unique equilibrium distribution of a given profile. 
This rule gives rise to a second kind of game, a \emph{revelation game}, in which the strategy of each donor consists of her reported utility function and contribution, and the collective outcome is the corresponding equilibrium distribution.

Our second set of results shows that \edr{} exhibits remarkable axiomatic properties:
\begin{itemize}
\item \emph{Group-strategyproofness}: agents and coalitions thereof are never better off by misrepresenting their preferences and 
are strictly better off by contributing more money (\Cref{thm:strategyproof}).
In particular, truthful reporting is a dominant strategy for all players in the revelation game. Moreover, as the equilibrium distribution maximizes Nash welfare, this game implements the Nash welfare function in dominant strategies.

\item \emph{Preference-monotonicity}: the amount donated to a charity can only increase when agents increase their valuation for the charity.
\item \emph{Contribution-monotonicity}: the amount donated to a charity can only increase when agents increase their contributions.
\end{itemize}

We further show that equilibrium distributions are the limit distributions of natural spending dynamics based on best responses. This can be leveraged in settings where a central infrastructure is unavailable or donors are reluctant to completely reveal their preferences. One could envision a scenario where each donor can announce her budget distribution on an online platform without making a commitment, allowing her to adapt it after viewing the donations made by others. We prove that when donors spend their money myopically optimally in each round, the relative overall distribution of donations converges to the equilibrium distribution. Hence, socially desirable outcomes can be attained even without a central infrastructure as long as charities are transparent about the donations they receive. This scenario also allows for occasional changes in the agents' preferences and contributions as the process keeps converging towards an equilibrium distribution of the current profile.\\

Apart from private charity, our results are also applicable to donation programs---prominent examples include \emph{AmazonSmile} and government programs (e.g., \emph{cinque per mille} run by the Italian Revenue Agency and \emph{mechanizm 1\%} in Poland).
In these programs, participants can redirect a portion of their payments (purchase price and income tax, respectively) to charitable organizations of their choice.\footnote{
These programs only allow each participant to choose exactly one charitable organization. However, as \citet{BBP+19a} argued, permitting them to indicate support for multiple organizations can increase the efficiency of the distribution.}
AmazonSmile ran from 2013 to 2023 and was used to allocate a total of \$400 million.
In 2022, a record \euro 510 
million were distributed via cinque per mille in Italy, and Poland recently increased the donatable quota of personal income tax from $1\%$ to $1.5 \%$.

In contrast to private charity, participants of donation programs do not have the option of taking their money out of the system, which means that the important issue lies in finding a desirable distribution of the contributions rather than in incentivizing the participants to donate in the first place.
A common criticism of the Polish program is that large organizations accumulate most of the donations, whereas some locally popular charities are left almost empty-handed.\footnote{See \polandurl \xspace and the references therein.}
The same phenomenon is well-known in healthcare, where organizations helping people with rare diseases find it difficult to attract donors. 
This issue is alleviated by our assumption of Leontief preferences, as illustrated by the following example.

\begin{example}
\label{exm:health}
Suppose there are ten donors, each donating $\$30$. Donor $i$ assigns value $2$ to a charity $A$ 
that supports patients with a common disease, 
and assigns value $1$ to a charity $B_i$
that supports patients with some rare diseases. 

If each donor is forced to select a single charity to donate to (as in the redirection programs mentioned above), then $A$ will receive all donations, and $B_1,\dots,B_{10}$ will receive none.
When donors can distribute their donations independently, without coordination, they will likely split their donations in the ratio $2:1$. As a result, $A$ will get $\$ 200$ and each $B_i$ will get only $\$ 10$. 

By contrast, our suggested mechanism will give $\$ 50$ to $A$ and $\$ 25$ to each $B_i$, which is the unique equilibrium. This distribution adequately reflects the donors' preferences, as each of them supports $A$ twice as much as $B_i$.
\end{example}

The remainder of this paper is structured as follows. After discussing related work in \Cref{sec: relwork}, we formally introduce our model in \Cref{sec: model}. \Cref{sec:equilibrium-distribution} lays the foundation for the proposed distribution rule by showing existence and uniqueness of equilibrium distributions via Nash welfare maximality as well as characterizing Pareto efficient distributions and pointing out connections to public good markets and Lindahl equilibrium. Subsequently, we define \edr as the rule that always returns the equilibrium distribution and examine its axiomatic properties in \Cref{sec:equilibrium-rule}.
In \Cref{sec: br-dynamics}, we explore natural spending dynamics that converge towards the equilibrium distribution.
The special case of Leontief utilities with binary weights 
allows for alternative characterizations of \edr, algorithms based on linear programming, and representations via a wide class of welfare functions; this
is covered in \Cref{sec:binary}. 
The paper concludes in \Cref{sec: conclusion} with a summary of our results and a brief discussion of alternative utility models such as linear, Cobb-Douglas, and leximin Leontief utilities (see also \Cref{app:otherutilityfunctions}). 
Almost all proofs are deferred to the Appendix.

\section{Related work}\label{sec: relwork}
\subsection{Public goods}
A well-studied problem related to the one we study in this paper is that of 
\emph{private provision of public goods} \citep[see, e.g.,][]{Samu54a,BBV86a,Vari94a,Falk96a,FFGW00a}. 
In this stream of research, each agent decides how much money she wants to contribute to funding a public good. Typically, this leads to under-provision of the public good in equilibrium, resulting in inefficient outcomes. 
In contrast, we consider agents who have already set aside a budget to support public charities, either voluntarily or compulsorily (as part of their taxes or payments to a central authority). 
This trivializes the problem of funding a \emph{single} public good. However, with two or more public goods, allocating donations among them may lead to inefficiencies, which our mechanism seeks to remedy.

Utilitarian welfare maximizing outcomes can be implemented by well-known strategyproof mechanisms such as the Vickrey-Clarke-Groves (VCG) mechanism.
However, VCG fails to be budget-balanced: 
it  collects money from the agents and has to `burn' that money in order to maintain strategyproofness.
By contrast, in our setting, the monetary contribution of each agent is fixed and independent of the agent's preferences. The entire contribution goes to charities approved by the agent and
the central issue is to distribute contributions fairly. 
As shown in \Cref{sec:sp}, strategyproofness can be achieved without imposing additional payments on the agents.

\subsection{Participatory budgeting}
A rapidly growing stream of research explores \emph{participatory budgeting} \citep[e.g.,][]{AzSh21a}, which allows citizens to jointly decide how the budget of a municipality should be spent in order to realize projects of public interest.
In contrast to charities, the projects considered in participatory budgeting typically come with a fixed cost (e.g., constructing a new bridge), and each project can be either fully funded or not at all.
Moreover, a standard assumption in participatory budgeting papers is that the money is owned by the municipality rather than by the agents themselves.

Perhaps the first work to consider charitable giving from a mechanism design perspective is due to 
\citet{CoSa04a,CoSa11a}. They let agents incentivize other agents to donate more by devising ``matching offers'', where a donation is made conditional on how much and to which charities other agents donate. They introduce an expressive bidding language for such offers and study the computational complexity of the resulting market clearing problem.

The work most closely related to ours is that of \citet{BBPS21a,BBP+19a}, who initiated the axiomatic study of donor coordination mechanisms. In their model, the utility of each donor is defined as the weighted \emph{sum} of contributions to charities, where the weights correspond to the donor's inherent utilities for a unit of contribution to each charity.
Under this assumption, the only efficient distribution in \Cref{exm:2donors-4charities} is to allocate the entire donation of $\$1000$ to charity~$C$, since this distribution gives the highest possible utility, $1000$, to all donors.
This distribution makes sense when the charities are perfect \emph{substitutes}, that is, they care for the same needs of the same people. In that case, when a donor assigns the same utility to several charities, she is completely indifferent to how money is distributed among these charities.
However, this distribution leaves charities $A$, $B$, and $D$ with no money at all, which may not be desirable when the charities are complementary, that is, they care for different needs or different people. 
By contrast, in our model of \emph{minimum-based} utilities, charities are perfect \emph{complements}: donors want their money to be evenly distributed among charities they like equally much. Fine-grained preferences over charities can be expressed by setting weights for Leontief utility functions. It can be argued that this assumption better reflects the spirit of charity by not leaving anyone behind. The modified definition of utility functions critically affects the nature of elementary concepts such as efficiency or strategyproofness and fundamentally changes the landscape of attractive mechanisms.

The main result of \citet{BBP+19a} shows that, in their model of linear utilities, the Nash product rule incentivizes agents to contribute their entire budget, even when attractive outside options are available. However, the Nash product rule fails to be strategyproof \citep{ABM20a} and violates simple monotonicity conditions \citep{BBPS21a}. 
In fact, a sweeping impossibility by \citet{BBPS21a} shows that, even in the simple case of binary valuations, no distribution rule that spends money on at least one approved charity of each agent can simultaneously satisfy efficiency and strategyproofness. This confirms a conjecture by \citet{BMS05a} and demonstrates the severe limitations of donor coordination with linear utilities. 
Interestingly, as we show in this paper, Leontief utilities allow for much more positive results.

\subsection{Nash product rule}
Originating from the \emph{Nash bargaining solution} \citep{Nash50b}, the Nash product rule can be interpreted as a tradeoff between maximizing utilitarian and egalitarian welfare, a recurring idea when it comes to finding efficient \emph{and} fair solutions.
When allocating divisible \emph{private} goods to agents with linear utility functions, the Nash product rule returns the set of all \emph{competitive equilibria from equal incomes} \citep{EiGa59a}; thus, it results in an efficient and \emph{envy-free} allocation \citep{Fole67a}. 
The connection between market equilibria and Nash welfare maximizers has been extended to various single-seller markets \citep{JaVa10a}. The most prominent of these are \emph{Fisher markets}, in which a set of divisible goods is available in limited quantities, and each buyer has a fixed budget at her disposal. The problem is to find equilibrium prices that clear the market while maximizing each buyer's utility. 

\citet{Eise61a} shows that when buyers have homogeneous utility functions (which include linear, Cobb-Douglas, and Leontief utility functions), market equilibria for Fisher markets can be found by maximizing the Nash product of individual utility functions \citep[see also][]{CoVa07a}.
While this resembles our \Cref{thm:unique-eq-distribution}, there are some important differences between Fisher markets and our public good markets. First of all, in contrast to Fisher markets, the relationship between equilibria and Nash welfare maximizers in public good markets does not extend to all homogeneous utility functions. In fact, for linear utility functions, all public good equilibrium distributions may be inefficient \citep[][Proposition~1]{BBP+19a}. Even for the special case of Leontief preferences, the relationship between equilibria and Nash welfare maximizers seems to be of a different nature. 
In contrast to our setting (see \Cref{sec:existuniquecomp}), Fisher market equilibria with Leontief preferences may involve irrational numbers and thus cannot be computed exactly \citep{CoVa04a}.\footnote{Interestingly, Nash welfare maximizers for Fisher markets with linear utilities are always rational-valued, whereas this is not true for our public good markets \citep[][Table~2]{BBP+19a}.} 
Moreover, mechanisms that maximize Nash welfare in private good settings do not share the desirable properties of \edr. For example, \citet{GZH+11a} show that the Nash product rule violates strategyproofness in a simple resource allocation setting with Leontief preferences.
To the best of our knowledge, there is no previous work on Leontief preferences in the context of \emph{public} goods.

\subsection{Many-to-many matchings}
A natural special case of our model is that of Leontief utilities with \emph{binary weights}, where agents only approve or disapprove charities, and the utility of each agent is given by the minimal amount transferred to any of her approved charities. Under the assumption that agents only contribute to charities they approve and that all individual contributions are equal, this can be interpreted as a (many-to-many) matching problem on a bipartite graph where agents (and their contributions) need to be assigned to charities with unlimited capacity.   
\citet{BoMo04a} propose a solution to such matching problems that maximizes egalitarian welfare of the charities (rather than the agents). The intriguing connection between these two types of egalitarianism is addressed in \Cref{sec:binary}. \citeauthor{BoMo04a} also show that their solution constitutes a competitive equilibrium from equal incomes (from the charity managers' point of view).

\section{The model}
\label{sec: model}

Let $N$ be a set of $n$ \emph{agents}. 
Each agent $i$ contributes an amount $C_i \ge 0$ of some divisible and homogeneous resource.
For every subset of agents $N'\subseteq N$, we denote $C_{N'} \coloneqq \sum_{i \in N'}C_i$. The sum of all contributions, $C_N$, is called the \emph{endowment}.
Further, consider a set $A$ of $m$ potential recipients of the contributions, which we refer to as \emph{charities}.

An \emph{individual distribution} for agent $i$ is a function~$\delta_i$ assigning a nonnegative real number to each charity, such that $\sum_{x\in A} \delta_i(x) = C_i$. 
A \emph{collective distribution} is a function~$\delta$ assigning a nonnegative real number to each charity, such that $\sum_{x\in A} \delta(x) = C_N$. 
Often, we will just say ``distribution'' when it is clear from the context whether it is an individual or a collective distribution.

Given a collective distribution $\delta$,
the support $\{x\in A\colon \delta(x) > 0\}$ of $\delta$ is denoted by $\supp(\delta)$. For a subset of charities $A'\subseteq A$, we define $\delta(A') \coloneqq \sum_{x \in A'}\delta(x)$ as the total amount allocated by $\delta$ to charities in $A'$.
The set of all collective distributions is denoted by $\Delta(C_N)$.

For every $i\in N$ and $x\in A$, there is a real number $v_{i,x} \geq 0$ that represents the value of charity $x$ to agent $i$. 
We assume that each agent $i$ has at least one charity $x$ for which $v_{i,x}>0$.
For every agent $i\in N$, we define $A_i \coloneqq \{x\in A \colon v_{i,x}>0 \}$ as the set of charities to which $i$ attributes a positive value. 

The utility that agent $i$ derives from a collective distribution $\delta$ is denoted by $u_i(\delta)$ and is given by the Leontief utility function:%

\begin{align*}
u_i(\delta) = \min_{x\in A_i}\frac{\delta(x)}{v_{i,x}}\text.
\end{align*}
Note that, 
for every charity $x\in A$ and every agent $i\in N$,
\begin{align*}
\delta(x)\geq v_{i,x}\cdot u_i(\delta).
\end{align*}

If all $v_{i,x}$ are in $\{0,1\}$, 
we have Leontief utilities with \emph{binary weights}.
A \emph{profile} $P$ consists of $\{C_i\}_{i\in N}$ and $\{v_{i,x}\}_{i\in N,\, x\in A}$.
Throughout this paper, agents with contribution zero do not have any influence on the outcome and can thus be treated as agents who choose not to participate in the mechanism.

Finally, a \emph{distribution rule} $f$ maps every profile of agents' utility functions to a collective distribution $\delta\in \Delta(C_N)$.

\section{Equilibrium distributions}
\label{sec:equilibrium-distribution}

The endowment to be distributed consists of the contributions of individual agents. In order to formalize which distributions are in equilibrium, we therefore need to define how (collective) distributions can be decomposed into individual distributions. 
\begin{definition}[Decomposition]\label{def:decomposition}
A \emph{decomposition} of a collective distribution $\delta$ is a vector of individual distributions $(\delta_i)_{i \in N}$ with
\begin{align}
\label{eq:dec-dx}
&\sum_{i \in N}\delta_i(x) = \delta(x) && \text{ for all } x\in A;
\\
\label{eq:dec-ci-A}
&\sum_{x \in A}\delta_i(x) = C_i && \text{ for all } i\in N.
\end{align}
\end{definition}
Clearly, each collective distribution admits at least one decomposition.
We aim for a decomposition in which no agent can increase her utility by changing $\delta_i$, given $C_i$ and the distributions $\delta_j$ for $j\neq i$. In other words, we look for a pure strategy Nash equilibrium of the game in which the strategy space of each agent $i$ is the set of individual distributions $\delta_i$ satisfying~\eqref{eq:dec-ci-A}.

\begin{samepage}
\begin{definition}[Equilibrium distribution]
\label{def:eqdistribution}

A vector of individual distributions $(\delta_i)_{i \in N}$ is \emph{in equilibrium} if for every agent $j$
and for every individual distribution  $\delta_j'$ satisfying 
$\sum_{x \in A}\delta_j'(x) = C_j$, 
\begin{align*}
u_j\left(\sum_{i \in N}\delta_i\right) \geq u_j\left(\sum_{i \in N}\delta_i-\delta_j+\delta_j'\right).
\end{align*}

A collective distribution is an \emph{equilibrium distribution} if it admits a decomposition into individual distributions that are in equilibrium.
\end{definition}
\end{samepage}

Given a vector of individual distributions $(\delta_i)_{i \in N}$, the collective distribution $\delta$ is understood to be $\delta\coloneqq\sum_{i \in N}\delta_i$.
 
The present section is devoted to establishing that every profile admits a \emph{unique} equilibrium distribution (\Cref{thm:unique-eq-distribution}), which also maximizes Nash welfare. 
All proofs are deferred to \Cref{app:equilibriumsection}.

Once we have established uniqueness, we can define the \emph{equilibrium distribution rule} as the rule that selects for each profile the corresponding equilibrium distribution.
In \Cref{sec:equilibrium-rule}, we will prove that this rule satisfies desirable strategic and monotonicity properties.

\subsection{Critical charities}
The following definition will be used to characterize the Nash equilibria of the distribution game.
Given a collective distribution $\delta$, 
we define the set of agent $i$'s \emph{critical charities}
\begin{align*}
T_{\delta,i} \coloneqq 
\arg \min_{x\in A_i} \frac{\delta(x)}{v_{i,x}}.
\end{align*}

Each charity $x\in T_{\delta,i}$ is critical for agent $i$ in the sense that the utility of $i$ would decrease if the amount allocated to $x$ were to decrease.
Every agent has at least one critical charity.
For every agent $i$ and charity $x$ such that either $v_{i,x} > 0$ or $\delta(x)>0$, the following equivalences hold:
\begin{equation}
\label{eq:critical-implications}
\begin{split}
x\in T_{\delta,i} &~~~\Leftrightarrow~~~ 
\delta(x) = v_{i,x}\cdot u_i(\delta);
\\
x\not\in T_{\delta,i} &~~~\Leftrightarrow~~~
\delta(x) >   v_{i,x}\cdot u_i(\delta).
\end{split}
\end{equation}
For every group of agents $N'\subseteq N$, we denote by $T_{\delta, N'}$ the set of charities critical to at least one member of $N'$.

The following lemma shows that a vector of individual distributions is in equilibrium if and only if each agent contributes only to her critical charities.
\begin{restatable}{lemma}{eqiffcritical}
\label{lem:eq-iff-critical}
A vector of individual distributions $(\delta_i)_{i \in N}$ is in equilibrium 
if and only if 
$\delta_i(x)=0$ for every charity $x\not\in T_{\delta,i}$.
Equivalently, it has to satisfy the following equality instead of \eqref{eq:dec-ci-A} in \Cref{def:decomposition}:
\begin{align}
\label{eq:dec-ci-critical}
&\sum_{x \in T_{\delta,i}}\delta_i(x) = C_i && \text{ for all } i\in N.
\end{align}
\end{restatable}

\Cref{lem:eq-iff-critical} implies that,
in an equilibrium distribution, every charity that receives a positive amount is critical for at least one agent.

Another implication of
\Cref{lem:eq-iff-critical} 
is that equilibrium distributions are stable even with respect to deviations by coalitions.
\begin{corollary}
\label{rem:strong-eq}
In equilibrium,
no \emph{group of agents} can deviate without making at least one of its members worse off. 
\end{corollary}
This is because \emph{any} deviation decreases the contribution to a critical charity of at least one group member. 
This equilibrium notion is slightly stronger than \emph{strong equilibrium} by \citet{Auma59a}. 

\subsection{Efficient distributions}

One of the main objectives of a centralized distribution rule is economic efficiency.

\begin{definition}[Efficiency]
Given a profile $P$, a distribution $\delta \in \Delta(C_N)$ is \emph{(Pareto) efficient} if there does not exist another distribution $\delta' \in \Delta(C_N)$ that \emph{(Pareto) dominates}~$\delta$, i.e., $u_i(\delta') \ge u_i(\delta)$ for all $i \in N$ and $u_i(\delta') > u_i(\delta)$ for at least one $i \in N$. 
A distribution rule is efficient if it returns an efficient distribution for every profile $P$.  
\end{definition}

\begin{observation}
\label{cor:equilibrium-is-efficient}
\Cref{rem:strong-eq} implies that every equilibrium distribution is efficient, since any Pareto improvement yields a beneficial deviation for the group of all agents.
\end{observation}

The following lemma characterizes efficient collective distributions.

\begin{restatable}{lemma}{efficientiffcritical}
\label{lem:efficient-iff-critical}
A collective distribution $\delta$ is efficient if and only if every charity  $x\in \supp(\delta)$ is critical for some agent. 
\end{restatable}

The next lemma shows that
every efficient utility vector is generated by at most one collective distribution.
\begin{restatable}{lemma}{efficientunique}
\label{lem:efficientunique}
Let $\delta$ and $\delta'$ be efficient distributions inducing the same utility vector, that is, $u_i(\delta)=u_i(\delta')$ for all $i \in N$. Then, $\delta=\delta'$.
\end{restatable}

Consequently, an efficient distribution rule is still well-defined when mapping a profile to an attainable utility vector instead of a distribution.

The following example illustrates the definitions and results we have so far. 

\begin{example}
Suppose there are three charities and two agents with the following valuations and contributions.
\[
\begin{array}{cccc@{\hskip 2em}c}
    \toprule
          & a & b & c & C_i \\ 
         \midrule
         v_1       & 1 & 1 & 0 & 1\\
         v_2       & 1 & 2 & 1 & 3\\
        \bottomrule
\end{array}
\]
Then, the unique equilibrium distribution is $\delta^*=(1,2,1)$ with $u_1(\delta^*)=u_2(\delta^*)=1$. Thus, charity $a$ is critical for Agent 1 while charity $b$ is not. All three charities are critical for Agent 2. $\delta^*$ can be decomposed into individual distributions $(1,0,0)+(0,2,1)$, which are in equilibrium. 

Not all decompositions of $\delta^*$ are in equilibrium.
For example, the decomposition $(0,1,0)+(1,1,1)$ allows Agent 1 to gain more utility by transferring some contribution from $b$ to $a$.

Apart from $\delta^*$, another efficient distribution is $\delta=(2,2,0)$, as $a$ and $b$ are critical for Agent 1
and $c$ is critical for Agent 2.
However, $0.5 \, \delta^* + 0.5 \, \delta=(1.5,2,0.5)$ is not efficient because $b$ is not critical for any agent. This shows that the set of efficient distributions fails to be convex, just like in the case of linear utilities \citep[see][]{BMS05a}.
\end{example}

\subsection{Existence, uniqueness, and computation}
\label{sec:existuniquecomp}
A common way to obtain efficient distributions is to maximize a welfare function. Formally, for any strictly increasing function $g$ on $\mathbb{R}_{\geq 0}$, 
we say that a collective distribution~$\delta$ is \emph{$g$-welfare-maximizing} if it maximizes the weighted sum $\sum_{i\in N} C_i\cdot g( u_i(\delta))$. Clearly, any such distribution is efficient.
Whenever $g$ is strictly concave, there is a \emph{unique} $g$-welfare-maximizing distribution; the straightforward proof is given in \Cref{sub:fwelfare}.

We focus on the special case in which $g$ is the logarithm function. 
The \emph{Nash welfare} of a distribution $\delta$ is defined as the sum of logarithms of the agents' utilities, weighted by their contributions:
\begin{align*}
\Nash(\delta) \coloneqq \sum_{i\in N} C_i\cdot \log u_i(\delta).
\end{align*}
The \emph{Nash rule} selects a collective distribution $\delta$ that maximizes $\Nash(\cdot)$ or, equivalently, the weighted product of the agents' utilities $\prod_{i \in N}u_i^{C_i}$ (with the convention that $0 \log 0=0$ and $0^0=1$).
In \Cref{app:proofuniqueeq}, we show that a collective distribution is an equilibrium distribution if and only if it maximizes Nash welfare. This implies the existence and uniqueness of equilibrium distributions.\footnote{An alternative \emph{non-constructive existence proof} can be obtained by leveraging a theorem by \citet{Rose65a}, who shows the existence of pure equilibria in $n$-player games 
with bounded, closed, and convex strategy sets when the payoff function is continuous in strategy profiles and concave in individual strategies. In fact, this already follows from \citeauthor{Debr52a}'s social equilibrium existence theorem \citep{Debr52a}.} 

\begin{restatable}{theorem}{uniqueeqdistribution}
\label{thm:unique-eq-distribution}
\label{thm:effcomputatibilityeq} 
Every profile admits a \emph{unique} equilibrium distribution. This distribution maximizes Nash welfare. 
\end{restatable}

The equilibrium distribution maximizes a weighted sum of logarithms and can thus be approximated arbitrarily well by considering the corresponding convex optimization problem. 
For linear utilities, \citet{BBP+19a} show that it is impossible to compute the Nash-optimal distribution exactly, even for binary valuations, since this distribution may involve irrational numbers.\footnote{In Fisher markets with linear utility functions, there exists a rational-valued equilibrium distribution, which can be computed \emph{exactly} in polynomial time via the ellipsoid algorithm and Diophantine approximation \citep{Jain07a,Vazi12a}.
}
By contrast, for Leontief utilities, the Nash-optimal distribution is rational whenever the agents' valuations and contributions are rational. This is the case because, given the sets of critical charities for each agent, the equilibrium distribution can be computed using linear programming. 
In an earlier version of this paper, we left open the question of whether these sets and thus the equilibrium distribution can be computed exactly in polynomial time.
In recent follow-up work, \citet{BFGS25a} answered this question by providing a pseudo-polynomial time algorithm. If all contributions are identical, the running time of their algorithm is even polynomial in the binary encoding length of the input. In the special case of binary weights, we show that the equilibrium distribution can be computed using a polynomial number of linear programs; see \Cref{sec:binary}.

\subsection{Public good markets}

We defined equilibrium distributions via Nash equilibria that arise in a strategic game where agents' strategies correspond to their individual distributions.
One can also see them as market equilibria in a stylized public good market. 
Consider a market in which each charity corresponds to a divisible public good in unlimited supply. 
Each agent has a budget of $C_i$ to spend on public goods and preferences on how the endowment $C_N$ should be distributed among the goods. Each public good is available at the same unit price.
A market equilibrium is a vector of individual distributions such that the distribution of each agent maximizes the agent's utility subject to her budget constraints.%
\footnote{This observation holds for general utility functions, not just Leontief.}

In the above market, all agents pay the same prices, but purchase different bundles.
\citet{Lind19a} introduced a ``dual'' market notion, in which each agent pays a personalized price, but purchases the same (public) bundle.
The following definition is adapted from Definition~5 by \citet{FGM16a}.
A distribution $\delta$ is a \emph{Lindahl equilibrium} if there exist personal price functions $p_1,\ldots,p_n\in \mathbb{R}_{\geq 0}^A$ such that the following two conditions hold.
\begin{enumerate}
     \item 
     \label{Lindahlcond1}
     For every agent $i \in N$, $\delta \in \arg \max_{z \in \mathbb{R}^A_{\geq 0}}u_i(z)$ subject to $\sum_{x \in A}p_{i}(x)z(x)=C_i$
     (i.e., $\delta$ maximizes the utility of $i$ among all collective bundles that $i$ can afford);
    \item 
    \label{Lindahlcond2} 
    For every good $x \in \supp(\delta)$, $\sum_{i \in N}p_{i}(x)=1$ 
    (i.e., the agents collectively pay $1$ for each unit of $x$);
    and for every good $x \not \in \supp(\delta)$, $\sum_{i \in N}p_{i}(x)\leq 1$
    (i.e., the collective price of an unproduced good is not higher than that of a produced good).
\end{enumerate}
Note that the maximization in the first case is over all $\mathbb{R}_{\geq 0}^A$ and not just over $\Delta(C_N)$.

 Compared to the definition by \citet{FGM16a}, we turned the inequality from the first condition into an equality, as agents are incentivized to spend their entire contributions for every utility function considered in this paper.
Furthermore, \citet{FGM16a} require $\delta$ to maximize a hypothetical profit of the form $\sum_{x \in A}(\sum_{i \in N}p_{i}(x)-1)\cdot z(x)$ where $z \in \mathbb{R}^A_{\geq 0}$ in the second condition. This is equivalent to our constraints on the sum of prices for every charity.\footnote{Additionally, \citet{FGM16a} do not require a Lindahl equilibrium to be a feasible distribution but only a nonnegative vector. This allows for additional (not sensible) solutions where arbitrary amounts of money are allocated to charities that are zero-valued by all agents (set individual prices for those charities to zero), but does not further expand the set of Lindahl equilibria.}

\citet{FGM16a} and \citet{GuPe20a} have established relationships between Nash welfare and Lindahl equilibrium for certain types of utility functions.
For example, \citet[][Theorem 2]{FGM16a} show that, if all agents have the same budget $C_N/n$ and all utilities are \emph{scalar-separable} and \emph{non-satiating}, then 
a Lindahl equilibrium can be computed efficiently via convex programming by maximizing a concave potential function that generalizes Nash welfare.
We prove an analogous result for Leontief utilities (which are neither scalar-separable nor non-satiating) and non-uniform budgets.

\begin{restatable}{proposition}{Lindahleq}
\label{prop:Lindahleq}
With Leontief utilities,
a collective distribution is an equilibrium distribution if and only if it is a Lindahl equilibrium.
\end{restatable}
Note that, since the equilibrium distribution is unique, the proposition also implies that the Lindahl equilibrium is unique.

For general utility functions, the equivalence between equilibrium distributions and Lindahl equilibria does not hold (e.g., in \Cref{rem:cobbdouglas} for Cobb-Douglas utilities in \Cref{app:otherutilityfunctions}, the Lindahl equilibrium $\delta$ Pareto dominates the equilibrium distribution.)

\section{The equilibrium distribution rule}
\label{sec:equilibrium-rule}
Based on \Cref{thm:unique-eq-distribution}, we define the \emph{equilibrium distribution rule (\edr)}
as the mechanism that, for each profile, returns the unique equilibrium distribution for this profile. In this section, we investigate axiomatic properties of \edr. All proofs are deferred to \Cref{app:axioms}.

\subsection{Strategyproofness and participation}\label{sec:sp}
A distribution rule is \emph{group-strategyproof} if no coalition of agents can gain utility by misreporting their valuations or decreasing their contribution. This incentivizes truthful reports and allows for a correct estimation of agents' utilities under different distributions. 
Furthermore, a group-strategyproof rule ensures that every agent donates the maximal possible contribution, thereby guaranteeing maximal gains from coordination.

It turns out that no group of agents has an incentive to misreport their valuations under \edr.

\begin{restatable}{theorem}{strategyproof}
\label{thm:strategyproof}
\edr is group-strategyproof, i.e.,
for all distinct profiles $P,P'$ and groups $G \subseteq N$ with 
$C'_G\leq C_G$ and $v_{i,x}=v'_{i,x}$ for all $i\in N\setminus G$, $u_i(\edr(P))= u_i(\edr(P'))$ for all $i\in G$ or there is some $i\in G$ such that $u_i(\edr(P))> u_i(\edr(P'))$.
\end{restatable}

The proof of \Cref{thm:strategyproof} can be used to show that an agent receives \emph{strictly} more utility when she increases her contribution. Since \edr returns a Nash welfare maximizing distribution, we can give an explicit lower bound on the utility gain when an agent increases her contribution.  

\begin{restatable}{proposition}{participation}
\label{thm:participation}
Under \edr, if the contribution of some agent increases by $Z>0$, then her utility increases by a factor of at least $(C_N+Z)/C_N$.
\end{restatable}

We are not aware of other settings in which the Nash product rule is strategyproof. \Cref{thm:strategyproof} strongly relies on Leontief preferences.\footnote{If, for example, preferences are given by Cobb-Douglas utility functions, the mechanism that always returns the equilibrium distribution can be manipulated (see \Cref{rem:cobbdouglas} in \Cref{app:otherutilityfunctions}).}

\subsection{Monotonicity conditions}

An important property from the perspective of charity managers is \emph{preference-monotonicity}, which requires that for every agent $i$ and charity $x \in A$, $\delta(x)$ weakly increases when $v_{i,x}$ increases.
In other words, a charity can only receive more donations when agents assign it a higher value.

For linear utilities, strategyproofness implies preference-monotonicity \citep{BBPS21a}. 
This does not hold for Leontief utilities, even when valuations are binary. 
Nevertheless, we still have the following.

\begin{restatable}{proposition}{prefmon}
\label{thm:prefmon}
\edr is preference-monotonic, i.e., for every two profiles $P$ and $P'$ which are identical except that $v_{i,x}' > v_{i,x}$ for one agent $i$ and one charity $x$,
$\edr(P')(x) \ge \edr(P)(x)$.
\end{restatable}

Again, this is in contrast to the linear utility case where the Nash product rule violates preference-monotonicity \citep{BBPS21a}.

To complement \Cref{thm:prefmon}, we observe two other effects of increasing the valuation $v_{i,x}$ of one agent~$i$ for one charity $x$:

\begin{itemize}
\item 
The maximal Nash welfare cannot increase; otherwise, the equilibrium distribution in the new profile would also have a higher Nash welfare than the equilibrium distribution of the original profile, with respect to the original valuations, contradicting \Cref{thm:unique-eq-distribution}.
However, the maximal Nash welfare can remain constant if $x$ is not among agent $i$'s critical charities.
\item 
Similarly, the utility of agent $i$ under the equilibrium distribution cannot increase: if agent $i$'s utility with the new valuation is larger under the new equilibrium distribution, this implies that her utility with the original valuation is also larger in the new equilibrium distribution and thus, there exists a beneficial manipulation (reporting exactly that new valuation instead). This would contradict strategyproofness of \edr (\Cref{thm:strategyproof}). 
However, agent $i$'s utility might remain constant if $x$ is not among her critical charities.
\end{itemize}

For some applications, it is desirable that increased contributions do not result in the redistribution of funds that have already been allocated. For example, if agents arrive over time or increase their contributions over time, ideally, the mechanism only needs to take care of the additional contributions.
This would allow a deployment of the mechanism as an incremental process in which charities can make immediate use of the donations they receive. 

\begin{restatable}{proposition}{contributionmon}
\label{thm:contribution-monotonicity}
\edr is contribution-monotonic, i.e., for every two profiles $P$ and $P'$ where $P'$ can be obtained from $P$ by increasing the contribution of one agent (possibly from $0$), $\edr(P')(x) \ge \edr(P)(x)$ for all charities $x \in A$.  
\end{restatable}

\section{Spending dynamics converging to equilibrium}\label{sec: br-dynamics} 

Thus far, we have assumed the existence of a central authority that collects the preferences of all agents and then either distributes the endowment among the charities or recommends to each agent how to distribute her individual contribution. In this section, we show that equilibrium distributions can also be attained in multi-round processes without a central authority, simply by letting agents spend their contribution one after another in a myopically optimal way. 
Agents need not reveal their preferences explicitly, but they have to be able to observe the donations made in previous rounds.

To this end, we consider infinite processes in which agents repeatedly play best responses against the strategies of the other agents in previous rounds. We first analyze a redistribution dynamics where the endowment remains fixed, and agents can redistribute their contribution whenever it is their turn. It turns out that the distribution converges to the equilibrium distribution under a very mild condition on the sequence of agents. We then consider a continuous spending dynamics in which there is a constant flow of contributions from each agent (for example, when each donor~$i$ has set aside a monthly budget~$C_i$ to spend on charitable activities). We focus on the case of round-robin sequences and show that the collective distribution of the last $n$ rounds and, thus, the relative overall distribution, converges to the equilibrium distribution when agents can only observe the distribution given by the last $n-1$ rounds. 

Our convergence results can be leveraged to make statements in more flexible settings where the set of participating agents, as well as their preferences and contributions, can change over time. The finite number of donations that have been made up to a certain point will always be outweighed by the infinite number of donations that follow. Hence, even with occasional changes to the profile, the relative overall distribution keeps converging towards an equilibrium distribution of the current profile.

All proofs for this section are deferred to \Cref{app: dynamics-leontief}.

\subsection{Redistribution dynamics}
\label{sub:redistribution-dynamics}

Let us first consider a dynamics in which the endowment remains fixed and agents repeatedly redistribute their contributions after observing the current collective distribution. 

Formally, denote by $\delta^*$ the equilibrium distribution and by $\deltat$ the distribution at round~$t$ (along with its associated decomposition), e.g., $\deltat[0]$ equals the null vector as no agent $i \in N$ has yet distributed her contribution $C_i$. 
In each round~$t$, allow one agent $i_t$ to (re-)distribute her entire contribution in such a way that her utility is maximized for the new distribution $\deltat[t+1]$, i.e.,
\begin{align*}
\delta^\mathit{best}_{i_t} &\coloneqq \arg\max_{\delta_{i_t} \in \Delta(C_{i_t})}u_{i_t}\left(\delta_{i_t} + \sum_{j \neq i_t}\deltat_j\right);
\\
\deltat[t+1] &\coloneqq \delta^\mathit{best}_{i_t} + \sum_{j \neq i_t}\deltat_j.
\end{align*}

\begin{restatable}{lemma}{uniquebestresponse}
\label{lem:uniquebestresponse}
For every round $t$ and agent $i_t$,
there is a unique best response $\delta^\mathit{best}_{i_t}$.
\end{restatable}

Before turning to the main result on the convergence of the dynamics, we illustrate the process using the initial example discussed in the introduction.
\begin{example}
Consider \Cref{exm:2donors-4charities} and suppose Donor 2 is the first in the sequence. Her best response, given the initial distribution $(0,0,0,0)$, is to split her donation of \$100 between $C$ and $D$, so the distribution becomes $(0,0,50,50)$.
	Next, Donor 1 plays a best response, which splits the donation of \$900 unequally, giving $316.\overline{6}$ to $A$, $316.\overline{6}$ to $B$, and $266.\overline{6}$ to $C$. The distribution becomes $(316.\overline{6},316.\overline{6},316.\overline{6},50)$.
	Then, Donor 2 plays a best response, which moves all of her donation to $D$. 
	The distribution becomes $(316.\overline{6},316.\overline{6},266.\overline{6},100)$.
	Finally, Donor~1 plays a best response, which moves $16.\overline{6}$ from each of $A$ and $B$ to $C$. The distribution then becomes $(300,300,300,100)$, which equals the equilibrium distribution. Consequently, agents' individual distributions form a Nash equilibrium.
In this particular case, the equilibrium distribution is attained after a finite number of rounds, but in general, we can only prove convergence in the limit.
\end{example}

Such dynamics are related to the ones studied in nonatomic congestion games \citep[see, e.g.,][]{Milc04a}. At a high level, each donor can be interpreted as a population of mass $C_i$ where each (infinitesimally small) agent or ``dollar'' wants to move to a charity $x$ with the smallest ratio $\delta(x)/v_{i,x}$, where $\delta(x)$ corresponds to the current mass of dollars at~$x$. While equilibria coincide (if no dollar wants to move, no donor has an incentive to redistribute her budget and vice versa), best responses of a donor are more complex as they can change the mass on more than two charities at once. Still, these connections turn out to be very useful.

We define a real-valued function $\Phi$ on the set of strategy vectors such that, whenever some agent deviates to a best response, $\Phi$ strictly increases.%
\footnote{
The definition of $\Phi$ is inspired by the definition of \emph{ordinal potential functions}, which  were originally introduced to prove the existence of pure strategy Nash equilibria in congestion games \citep{Rose73a,MoSh96a} and have since then been widely used to prove convergence to equilibrium \citep[e.g.,][]{Milc96b,Milc00a,Milc04a}.
However, an ordinal potential function increases whenever a player plays a \emph{better} response, whereas our $\Phi$ increases only when a player plays a \emph{best} response. 
In fact, the game we consider does not admit an ordinal potential function, as there can be cycles of better responses.

\citet{Voor00a} defines a \emph{best-response potential function}, which is maximized whenever a player plays a best response; our $\Phi$ increases when a player plays a best response, but we do not know if it is maximized.

Note that Nash welfare need \emph{not} monotonically increase for best-response sequences.
}

\begin{align}
\label{eq:potential}
    \Phi(\delta_1,\ldots,\delta_n)\coloneqq\sum_{i \in N}\sum_{x \in A_i}\delta_i(x)\log \left(\frac{v_{i,x}}{\delta(x)}\right)
\end{align}
Note that $\Phi$ is well-defined, as $\delta(x)=0$ implies $\delta_i(x)=0$ for all $i\in N$, and $x\in A_i$ implies $v_{i,x}>0$.

Observe that an agent's best response can be described by the following continuous process: 
as long as the agent spends a positive amount on a non-critical charity, transfer money from such a charity to all critical charities equally until either \emph{(i)} at least one more charity becomes critical or \emph{(ii)} the agent no longer spends a positive amount on a non-critical charity.
This process can be interpreted as a sequence of transfers, where each transfer of amount $\epsilon>0$ goes from a charity $x$ with higher weighted distribution to a charity $y$ $\left(\frac{\delta(x)}{v_{i,x}} > \frac{\delta(y)}{v_{i,y}}\right)$ such that after the transfer, the weighted distribution of the former charity remains at least as high as that of the latter: $\frac{\delta(x)-\epsilon}{v_{i,x}} \geq \frac{\delta(y)+\epsilon}{v_{i,y}}$.
The following lemma proves that each such transfer strictly increases $\Phi$.

\begin{restatable}{lemma}{lempotentialpairwise}
\label{lem:potential-pairwise}
Let $(\delta_1,\ldots,\delta_i,\ldots,\delta_n)$ be any vector of individual distributions,
$i\in N$ any agent,
and $\delta_i'$ a distribution derived from $\delta_i$ when agent $i$ moves an amount $\epsilon>0$ from some charity $x$ to another charity $y$ with
$\frac{\delta(x)-\epsilon}{v_{i,x}} \geq \frac{\delta(y)+\epsilon}{v_{i,y}}$.
Then, the move strictly increases $\Phi$, that is,
$\Phi(\delta_1,\ldots,\delta'_i,\ldots,\delta_n)>\Phi(\delta_1,\ldots,\delta_i,\ldots,\delta_n)$.
\end{restatable}

As every best response is a sum of pairwise transfers as in \Cref{lem:potential-pairwise},%
\footnote{
\label{ftn:better-response}
This argument does \emph{not} hold for any better response. For example, suppose an agent values three charities equally, the current collective distribution is $3,7,9$, and the agent changes it to $4,5,10$. This better response (which is not a best response) cannot be derived from transfers as in \Cref{lem:potential-pairwise}.
}
we obtain the following lemma as a corollary.
\begin{restatable}{lemma}{lempotential}
\label{lem:potential}
    For any best-response sequence $\mathcal{S}$, it holds that $\Phi(\deltat[t+1])>\Phi(\deltat[t])$ for all $t$.
\end{restatable}

\begin{restatable}{theorem}{dynamics}
\label{thm:leon-dynamics}
    Given a profile $P$, let $\mathcal{S}=(i_0,i_1,i_2,\dots)$ be an infinite sequence of agents updating their individual distributions via best responses such that there is a bound $K \in \mathbb{N}$ on the maximal number of rounds an agent has to wait until she is allowed to redistribute.
    Then, the redistribution dynamics converges to the equilibrium distribution, i.e.,
    $\lim_{t \to \infty} \delta^t=\delta^*$. 
\end{restatable}

For binary Leontief utilities, the function $\Phi$ from \eqref{eq:potential} simplifies to $\Phi(\delta_1,\ldots,\delta_n)=-\sum_{x \in A}\delta(x) \log(\delta(x))$.

Also, with binary Leontief utilities, $\widehat{\Phi}(\delta) \coloneqq -\sum \lvert\delta(x) - \delta(y)\rvert$, where the sum is taken over all (unordered) pairs of distinct charities $x,y\in A$, is an alternative choice for $\Phi$; it also increases via best responses. 
Moreover, $\widehat{\Phi}$ has the advantage that it increases linearly in the redistributed amounts:
if an agent redistributes some amount, then the function changes by exactly that amount. As $\widehat{\Phi}$ is upper-bounded, we know that the total amount of redistributed contributions is bounded.

It can then be shown that \Cref{thm:leon-dynamics} holds for \emph{any} best-response sequence~$\mathcal{S}$ in which each agent appears infinitely often (we do not need the uniform upper bound $K$): if an agent wants to redistribute an amount $d$ at round $t$, then $\widehat{\Phi}$ will increase by at least $d$ after that agent redistributes. So we again get convergence. 

\subsection{Round-robin spending dynamics}
\label{sub:spending-dynamics}

Let us now move on to a model in which there is a constant flow of donations, and each agent repeatedly donates her contribution $C_i$ when it is her turn. To this end, fix some order of the agents (say, $1,2,\dots, n$) and denote by $\delta^t_i(x)$ the total amount of contributions of agent $i$ to charity $x$ until round $t$. At each round $t \ge 0$, agent $i_t=1+ (t \mod n)$
donates $C_{i_t}$ in such a way that her utility is maximized with respect to the previous donation of every other agent, i.e.,
\begin{align*}
\delta^{\mathit{best}}_{t} &\coloneqq \arg\max_{\delta_{i_t} \in \Delta(C_{i_t})}u_{i_t}\left(\delta_{i_t} + \sum_{t-n<s<t}\delta^{\mathit{best}}_{s}\right);
\\
\deltat[t+1] &\coloneqq \delta^t + \delta^{\mathit{best}}_{t};
\\
\deltat[t+1]_{i_t} &\coloneqq \deltat_{i_t}+\delta^{\mathit{best}}_{t},
\end{align*}
where the distribution of the contribution of agent $i_t$ in round $t$ is denoted by $\delta^{\mathit{best}}_{t}$.\footnote{We here assume that the ``observation window'' of each agent is given by the last $n-1$ rounds. Computer simulations suggest that convergence also holds for larger observation windows.}

\Cref{lem:uniquebestresponse} still applies: each agent's best response is unique. To compare $\delta^t$ with the equilibrium distribution $\delta^*$ (where each agent only contributed once), we scale $\delta^t_i$ by the number of donations of agent $i$ until round $t$, which equals $\lfloor (t+n-i)/n \rfloor$. 

\begin{example}
To illustrate the process, consider \Cref{exm:2donors-4charities} from the introduction for the sequence $(2,1,2,1,\dots)$. First, Donor 2 splits her donation of \$100 between $C$ and $D$, resulting in $\delta^1=(0,0,50,50)$. Next, Donor 1 plays a best response, which splits the donation of \$900 unequally, giving $316.\overline{6}$ to $A$, $316.\overline{6}$ to $B$ and $266.\overline{6}$ to $C$ leading to $\delta^2=(316.\overline{6},316.\overline{6},316.\overline{6},50)$. Then, Donor 2 donates another \$100 to $D$ under her best response. The overall distribution becomes $\delta^3=(316.\overline{6},316.\overline{6},316.\overline{6},150)$. It is straightforward to see that from now on, Donor 1 will always split her contribution equally on $A$, $B$, and $C$, whereas Donor 2 will only donate to $D$. Thus, $\lim_{t \to \infty}2\delta^t_1/t=(300,300,300,0)$ and $\lim_{t \to \infty}2\delta^t_2/t=(0,0,0,100)$, showing convergence to the equilibrium distribution. 
\end{example}

\begin{restatable}{theorem}{contspenddynamics}
\label{thm:contspenddynamics}
Given a profile $P$, the continuous round-robin spending dynamics converges to the equilibrium distribution, i.e.,
\begin{align*}
\lim_{t \to \infty}\sum_{i \in N}\frac{1}{\lfloor (t+n-i)/n \rfloor} \delta^t_i=\delta^*\text.
\end{align*}
\end{restatable}

\begin{remark}\label{rem: sepadd-dynamics}
    In the spirit of \Cref{prop:sepaddutils,prop:otherutilsthanCD}, the convergence results also apply to other utility functions (e.g., Cobb-Douglas), as not only equilibrium distributions but also best responses coincide. 
\end{remark}

\begin{remark}
An interesting question is which profiles allow for the existence of sequences such that the dynamics not only converges to but even reaches the equilibrium distribution, and more generally, how optimal sequences can be found. 
\end{remark}

\section{Leontief utilities with binary weights}
\label{sec:binary}

In this section, we consider the special case of binary Leontief weights, i.e., $v_{i,x} \in \{0,1\}$ for all agents $i \in N$ and charities $x \in A$.
Equivalently, each agent $i$ has a non-empty set of \emph{approved charities} $A_i\subseteq A$ and her utility from a distribution $\delta$ is 
\begin{align*}
u_i(\delta) = \min_{x\in A_i} \delta(x).
\end{align*}
For each charity $x\in A$, we denote by $N_x\subseteq N$ the set of agents who approve charity~$x$.
For a subset of agents $N'\subseteq N$, we denote $A_{N'} \coloneqq \bigcup_{i \in N'} A_i$ as the set of charities approved by at least one member of $N'$.
Note that, for every charity $x\in A$ and every agent $i\in N_x$,
\begin{align}
\label{eq:min-based}
\delta(x)\geq u_i(\delta).
\end{align}
Binary weights allow for further insights into the structure of the equilibrium distribution, which in turn yield new interpretations and additional properties of \edr.

For linear utilities with binary weights, a vector of individual distributions is in equilibrium if and only if each agent contributes only to charities she approves. For the resulting collective distribution, \citet{BBPS21a} refer to this axiom as \emph{decomposability}. 

\begin{definition}[Decomposable distribution]
\label{def: decomposability}
Given a profile with binary weights ($v_{i,x} \in \{0,1\}$),
a distribution $\delta$ is \emph{decomposable} if it has a decomposition $(\delta_i)_{i \in N}$ such that $\delta_i(x)=0$ for every charity
$x\not\in A_i$.
Equivalently, it has a decomposition satisfying the following, instead of \eqref{eq:dec-ci-A}:
\begin{align*}
&\sum_{x \in A_i}\delta_i(x) = C_i && \text{ for all } i\in N.
\end{align*}
\end{definition}

The equivalence of decomposable distributions and equilibrium distributions no longer holds for Leontief utilities: there are decomposable distributions that are not equilibrium distributions even when there is only one agent.

Nevertheless, decomposability can be used to establish two appealing alternative interpretations of \edr for binary weights.

\subsection{Egalitarianism for charities}
Motivated by \Cref{exm:2donors-4charities}, we aim at a rule that distributes money on the charities as equally as possible
while still respecting the preferences of the donors.
One rule that comes to mind selects a distribution that, among all decomposable distributions, maximizes the smallest amount allocated to a charity. Subject to this, it maximizes the second-smallest allocation to a charity, and so on. We define it formally using the \emph{leximin} relation.

\begin{definition}
\label{def:leximin}
Given two vectors $\mathbf{x}, \mathbf{y}$ of the same size, we say that \emph{$\mathbf{x}$ is leximin-higher than $\mathbf{y}$} (denoted $\mathbf{x} \succ_\mathit{lex} \mathbf{y}$) if the smallest value in $\mathbf{x}$ is
larger than the smallest value in $\mathbf{y}$; or the smallest values are equal, and the second-smallest 
value in $\mathbf{x}$ is
larger than the second-smallest value in $\mathbf{y}$; and so on. 
$\mathbf{x} \succeq_\mathit{lex} \mathbf{y}$ means that either $\mathbf{x} \succ_\mathit{lex} \mathbf{y}$ or 
the multiset of values in $\mathbf{x}$ is the same as that in $\mathbf{y}$.
\end{definition} 

\begin{definition}
The \emph{charity egalitarian rule} selects the distribution $\delta^{*}$ that, among all decomposable distributions,
maximizes the distribution vector by the leximin order, that is: $\delta^{*}\succeq_\mathit{lex} \delta$ for every decomposable distribution $\delta$.
\end{definition}

The leximin order on the closed and convex set of decomposable distributions is connected, every two vectors are comparable, and there exists a unique maximal element (otherwise, any convex combination of two different maximal elements would be leximin-higher than the maximal elements).
Therefore, the charity egalitarian rule selects a unique distribution and is well-defined.
We prove that the returned distribution is the equilibrium distribution, resulting in an alternative characterization of \edr for binary weights.

\begin{restatable}{proposition}{edrisprojegal}
\label{thm:edrisprojegal}
With binary weights,
the charity egalitarian rule and \edr are equivalent.
\end{restatable}

Remarkably, this new interpretation of \edr ignores the Leontief utilities of the agents and does not directly take into account the different contributions. Instead, they enter indirectly through the constraints induced by decomposability. 

\begin{observation}
\label{cor: lex-app}
When new approvals are added, the equilibrium distribution becomes (weakly) leximin-higher with respect to the vector of contributions on the charities.
\end{observation}
\begin{proof}
This follows directly from \Cref{thm:edrisprojegal}.
When new approvals are added, the decomposability constraints on the vector of charity distributions become weaker: every distribution that was decomposable before the addition remains decomposable after the addition. Therefore, the new leximin-maximal distribution is at least as high, in the leximin order, as the previous distribution.
\end{proof}

\Cref{thm:edrisprojegal} implies that \edr can be computed by solving the following program, with variables
$\delta_x$ for all $x\in A$ and $\delta_{i,x}$ for all $i\in N,\, x\in A$:
\begin{align*}
\operatorname{lex} \operatorname{max} \operatorname{min} &&\{ \delta_x\}_{x\in A}
\\
\text{subject to}&&\delta_x &= \sum_{i\in N} \delta_{i,x}
&& 
\text{for all } x\in A
\\
&&\sum_{x\in A_i}\delta_{i,x} &= C_i
&&
\text{for all } i\in N
\\
&&\delta_{i,x} &\ge 0
 &&\text{for all } i\in N,\, x\in A_i
\end{align*}
where ``$\operatorname{lex} \operatorname{max} \operatorname{min}$'' refers to finding a solution vector that is maximal in the leximin order subject to the constraints, and the second constraint represents decomposability.
It is well-known that such leximin optimization with $k$ objectives 
and linear constraints can be solved by a sequence of $k$ linear programs \citep[see, e.g.,][Sect.~5.3]{Ehrg05a}.

\begin{corollary}
With binary weights, 
the equilibrium distribution can be computed by solving at most $m$ linear programs.
\end{corollary}

\subsection{Egalitarianism for agents}
\label{subsec: interpretations}

While \edr is egalitarian from the point of view of the charities, one could also consider a rule that is egalitarian from the point of view of the agents. 
The \emph{conditional egalitarian rule} aims to balance the agents' utilities without disregarding their approvals. 
It selects the decomposable distribution that, among all decomposable distributions, maximizes the utility vector by the leximin order, that is: $\mathbf{u}(\delta^*)\succeq_\mathit{lex} \mathbf{u}(\delta)$ for every decomposable distribution $\delta$.

\begin{definition}
The \emph{conditional egalitarian rule} selects the distribution $\delta^{*}$ that, among all decomposable distributions,
maximizes the utility vector by the leximin order, that is: $(u_1(\delta^{*}),\dots,u_n(\delta^*))\succeq_\mathit{lex} (u_1(\delta),\dots,u_n(\delta))$ for every decomposable distribution $\delta$.
\end{definition}

It is well-defined for similar reasons as the charity egalitarian rule.

\begin{restatable}{proposition}{edrisceg}
\label{thm:edrisceg}
With binary weights, the conditional egalitarian rule and \edr are equivalent.
\end{restatable}

\Cref{thm:edrisceg} implies that the equilibrium distribution can be computed by solving the following program with variables
 $u_i$ for all $i\in N$ and $\delta_{i,x}$ for all $i\in N,\, x\in A_i$.
\begin{align*}
\operatorname{lex} \operatorname{max} \operatorname{min} &&\{ u_i \}_{i\in N}
\\
\text{subject to}
&&u_i &\leq \sum_{j \in N} \delta_{j,x} 
&& 
\text{for all } i\in N,\, x\in A_i
\\
&&\sum_{x\in A_i}\delta_{i,x} &= C_i
&&
\text{for all } i\in N
\\
&&\delta_{i,x} &\ge 0
 &&\text{for all } i\in N,\, x\in A_i\\
&&u_i &\geq 0 &&\text{for all } i\in N
\end{align*}
Again, using standard algorithms for lexicographic max-min optimization, this program can be solved using at most $n$ linear programs.

Thus, we have three algorithms for computing the equilibrium distribution in the case of binary weights: one requires at most $m$ linear programs; one requires at most $n$ linear programs; and one requires a single convex (non-linear) program. It would be interesting to investigate which algorithms are most efficient in practice.

Note that, for \emph{general} Leontief utilities, equilibrium distributions may neither maximize the leximin vector of the charities nor of the agents.
\begin{example}
\label{exm:equilibrium-is-not-egalitarian}

Suppose there are three charities and two agents with the following valuations and contributions.
\[
\begin{array}{cccc@{\hskip 2em}c}
    \toprule
          & a & b & c & C_i \\ 
         \midrule
         v_1       & 1 & 2 & 0 & 30\\
         v_2       & 0 & 1 & 1 & 30\\
        \bottomrule
\end{array}
\]
The \emph{charity egalitarian rule} returns the leximin-maximal distribution for charities (subject to decomposability), which is $(20,20,20)$ with decomposition $(20,10,0)+(0,10,20)$.
It is not the equilibrium distribution since Agent 1 contributes to charity $a$, which is not critical.

The \emph{conditional egalitarian rule} returns the leximin-maximal distribution for agents (subject to decomposability), which is $(15,30,15)$, with utility vector $(15,15)$ and decomposition $(15,15,0)+(0,15,15)$.
It is also not the equilibrium distribution since Agent 2 contributes to charity $b$, which is not critical.

\edr returns $(12,24,24)$ with decomposition $(12,18,0)+(0,6,24)$ and utility vector $(12,24)$. 
\end{example}

\subsection{Welfare functions maximized by \edr}

Based on the observation that \edr coincides with both the Nash product rule and the conditional egalitarian rule for binary weights, a natural question to ask is which other welfare notions are maximized by \edr subject to decomposability.

For this, we take a closer look at \gwelfare (see \Cref{sec:existuniquecomp} and \Cref{sub:fwelfare}), but this time subject to decomposability. Clearly, every \gwelfare-maximizing distribution is efficient. Below, we prove that efficiency is retained even when maximizing among decomposable distributions.
\begin{restatable}{lemma}{fmaximizingisefficient}
\label{lem:fmaximizingisefficient}
Let $g$ be any strictly increasing function, and let $\delta$ be a distribution that maximizes \gwelfare among all decomposable distributions. Then, $\delta$ is unique and efficient.
\end{restatable}

Note that uniqueness holds only within the set of decomposable distributions; there might exist non-decomposable distributions with the same \gwelfare, as shown in the following example.
\begin{example} 
Let $g(x) = -x^{-1}$ (a strictly increasing function). Suppose there are two charities and two agents with the following valuations and contributions.
\[
\begin{array}{ccc@{\hskip 2em}c}
    \toprule
          & a & b & C_i \\ 
         \midrule
         v_1       & 1 & 0 & 2\\
         v_2       & 0 & 1 & 1\\
        \bottomrule
\end{array}
\]
Then, the unique decomposable distribution $(2,1)$ has \gwelfare $-2/2-1/1 = -2$, and the non-decomposable distribution $(1.5,1.5)$ also has \gwelfare $-2/1.5-1/1.5 = -2$.
\end{example}

The Nash product rule is often considered a compromise between maximizing utilitarian welfare ($\sum_{i \in N} C_i \cdot u_i$) and egalitarian welfare (maximizing the utility of the agent with the smallest utility; notice that the conditional egalitarian rule is a refinement). This can be seen when considering the family of \gwelfare functions $\sum_{i \in N}C_i \cdot \sgn(p) \cdot u^p$ for $p\neq 0$, where the limit $p \to 0$ corresponds to $\sum_{i \in N}C_i \cdot \log(u_i)$ and $p \to -\infty$ approaches egalitarian welfare \citep{Moul03a}. 

Interestingly, this family is characterized within the class of \gwelfare functions by one additional axiom called \emph{independence of the common utility scale} when demanding that each agent needs to receive positive utility \citep[see][]{Moul88a}. This axiom requires that the induced social welfare ordering does not change under scaling utilities. Assuming that the agents' utilities scale linearly (e.g., $u_i(2\delta)=2u_i(\delta)$), it is equivalent to demanding that the ordering, and thus also the mechanism that chooses one of its maximal elements, does not depend on the valuation of a unit of contribution as long as all agents agree on it.

The equivalence between conditional egalitarian welfare and Nash welfare extends to a larger class of \gwelfare functions. 
This is shown by the following theorem, proven in \Cref{app: equivalenceproof}.

\begin{restatable}{theorem}{maxfwelfareequivalence}
\label{thm: cd-is-fmaximal}
Let $g:\mathbb{R}_{\ge 0} \to \mathbb{R}\cup\{-\infty\}$ be a function that satisfies the following conditions:
{
\let\oldlabelenumi\labelenumi
\renewcommand{\labelenumi}{(\arabic{enumi})}
\begin{enumerate}
    \item $g$ is strictly increasing on $\mathbb{R}_{\ge 0}$ and differentiable on $\mathbb{R}_{>0}$, and
    \item $x g'(x)$ is non-increasing on $\mathbb{R}_{>0}$.
\end{enumerate}
\let\labelenumi\oldlabelenumi
}
Then, for any instance with binary Leontief valuations, the equilibrium distribution maximizes \gwelfare among all decomposable distributions.
\end{restatable}

Property (1) ensures that social welfare is indeed increasing when an individual's utility increases and small changes in individual utilities only cause small changes in the total social welfare. Property (2) implies that increasing utilities are discounted ``at least logarithmically'' when translated to welfare.

In particular, \Cref{thm: cd-is-fmaximal} holds for all \gwelfare functions $\sum_{i \in N}C_i \cdot \sgn(p) \cdot u^p$ with $p < 0$.
However, it ceases to hold when $p>0$, as the following proposition (whose proof is deferred to \Cref{app: p>0}) shows.

\begin{restatable}{proposition}{pwelfarelemma}\label{prop: p>0}
For each $p>0$, maximizing the \gwelfare with respect to $g(u)=u^p$ subject to decomposability does not always return the equilibrium distribution, even with binary Leontief valuations. 
\end{restatable}

\Cref{thm: cd-is-fmaximal} underscores the fact that \edr can be motivated not only from a game-theoretic and axiomatic point of view but also from a welfarist perspective.

\section{Discussion}
\label{sec: conclusion}

Under the assumption that donors' preferences can be modeled using Leontief utility functions, \edr is an exceptionally attractive rule for funding charitable organizations.
It is efficient, group-strategyproof, and monotonic in both preferences and contributions. Moreover, it can be computed via convex programming and returns the limit of natural spending dynamics. In the case of binary weights, \edr maximizes a wide range of welfare functions and can be computed via linear programming.
These results stand in sharp contrast to the previously studied case of linear utilities, where a far-reaching impossibility has shown the incompatibility of efficiency, strategyproofness, and a very weak form of fairness \citep{BBPS21a}. The literature in this stream of research has produced various rules such as the \emph{conditional utilitarian rule}, the \emph{Nash product rule}, the \emph{random priority rule}, or the \emph{sequential utilitarian rule}, which trade off these properties against one another \citep{BMS05a,Dudd15a,ABM20a,BBPS21a,BBP+19a}.

Equilibrium distributions can be interpreted as market equilibria for a pure public good market with unlimited supply. This perspective allows interesting comparisons to Fisher markets, arguably the simplest markets for divisible private goods. Equilibria in Fisher markets are connected to Nash welfare maximization under fairly general assumptions about individual utilities, whereas this connection appears to be more volatile in our public good markets. On the other hand, even for Leontief preferences, Fisher market equilibria cannot be computed exactly \citep{CoVa04a}, and the mechanism that returns the equilibrium is manipulable \citep{GZH+11a}.

An important question is to which extent our results carry over to other concave utility functions, which offer a natural middle ground between linear and Leontief utilities. 
\Cref{rem: sepadd-dynamics} as well as \Cref{prop:sepaddutils,prop:otherutilsthanCD} show that equilibrium existence, uniqueness, and convergence of the best-response-based spending dynamics also hold for other utility models. 
However, other axiomatic properties might break down, e.g., the equilibrium distribution may fail to be efficient (\Cref{rem:cobbdouglas}) and the rule that returns the equilibrium distribution is not strategyproof (see \Cref{sec:sp}). 

Leontief preferences can be refined by breaking ties between distributions lexicographically, similar to leximin utilities. More precisely, rather than only caring about the minimum of $\delta(x)/v_{i,x}$ for $x\in A$, agents can rank all distributions according to the leximin relation (\Cref{def:leximin}) among the vectors $(\delta(x)/v_{i,x})_{x\in A}$. Remarkably, all of our results for general Leontief valuations carry over to these utility functions by adapting the proofs accordingly. It should be noted, however, that lexicographic Leontief preferences are discontinuous.

\subsection*{Acknowledgements}
{\footnotesize

This material is based on work supported by the Deutsche Forschungsgemeinschaft under grants BR 2312/11-2 and BR 2312/12-1, by the Israel Science Foundation under grant numbers 712/20 and 1092/24, by the Singapore Ministry of Education under grant number MOE-T2EP20221-0001, and by an NUS Start-up Grant.
We are grateful to Florian Brandl for proposing the best response dynamics together with a proof idea, Igal Milchtaich for pointing out fruitful connections to nonatomic congestion games, and Ido Dagan for suggesting lexicographic Leontief preferences. We further thank 
editor Faruk Gul, the anonymous referees, 
Ronen Gradwohl, Ilan Kremer, Somdeb Lahiri, Hervé Moulin, Noam Nisan, Danisz Okulicz, Marcus Pivato, Clemens Puppe, Marek Pycia, Ella Segev, Vijay Vazirani, Rakesh Vohra
as well as the participants of the 
3rd Ariel Conference on the Political Economy of Public Policy (September 2022), 
the joint Microeconomics Seminar of ETH Zurich and the University of Zurich (March 2023),  
the Bar-Ilan University Computer Science Seminar (April 2023),
the Hebrew University of Jerusalem Econ-CS seminar (May 2023),
the Bar-Ilan University Game Theory seminar (June 2023), 
the 9th International Workshop on Computational Social Choice in Beersheba (July 2023), 
the 24th ACM Conference on Economics and Computation (July 2023), 
the KIT Conference on Voting Theory and Preference Aggregation  (October 2023), 
the Online Social Choice and Welfare Seminar (January 2024), 
the Second Vienna-Graz Workshop on (Computational) Social Choice (February 2024), 
the 17th Meeting of the Society for Social Choice and Welfare (July 2024),
the Equity and Access in Algorithms, Mechanisms, and Optimization Online Colloquium (May 2025), and
the Oxford AGT-Seminar (July 2025) for their insightful comments, stimulating discussions, and encouraging feedback.\par}

\appendix

\newpage
\section*{APPENDIX}

\section{Alternative utility models}\label{app:otherutilityfunctions}

\subsection{Cobb-Douglas utilities}
Given values for charities, the \emph{Cobb-Douglas utility function} of agent $i$ is defined as $u_i(\delta)=\prod_{x \in A} \delta(x)^{v_{i,x}}$.
\footnote{
Both Cobb-Douglas and \emph{binary} Leontief utility functions belong to the class of utility functions with \emph{constant elasticity of substitution} \citep[see, e.g.,][]{Vari92a,MWG95a}. 
Similarly to the case of maximizing the Nash product of individual utilities, maximizing a single Cobb-Douglas utility function is equivalent to maximizing $\sum_{x \in A} v_{i,x}\cdot \log(\delta(x))$.
}
It turns out that the equilibrium distribution is unaffected if the agents' Leontief utility functions are replaced with Cobb-Douglas utility functions for the same values for charities.

\begin{proposition}\label{prop:sepaddutils}
Given values $(v_{i,x})_{i\in N,\, x\in A}$, a vector of individual distributions is an equilibrium for Leontief utility functions if and only if it is an equilibrium for the corresponding Cobb-Douglas utility functions.
\end{proposition}

\begin{proof}
We show that \Cref{lem:eq-iff-critical}, with the same definition of critical charities  $T_{\delta,i} \coloneqq 
\arg \min_{x\in A_i} \frac{\delta(x)}{v_{i,x}}=\arg \max_{x\in A_i} \frac{v_{i,x}}{\delta(x)}$, 
also holds for Cobb-Douglas utilities $u_i(\delta)=\sum_{x \in A} v_{i,x}\cdot \log(\delta(x))$.
 
$\Rightarrow$: 
Given a vector of individual distributions $(\delta_i)_{i \in N}$, suppose some agent~$i$ (with Cobb-Douglas utilities) contributes to a charity $y\not\in T_{\delta,i}$. By definition of critical projects, 
$v_{i,y}/\delta(y) < v_{i,x} / \delta(x)$ for any $x \in T_{\delta,i}$.
Let agent $i$ move a sufficiently small amount $\epsilon$ from $\delta_i(y)$ to a charity $x \in T_{\delta,i}$; denote the resulting individual distribution by $\delta^\epsilon_i$. This move strictly increases $i$'s utility, since
\begin{align*}
    &\; \lim_{\epsilon \to 0} \frac{u_i(\delta-\delta_i+\delta^\epsilon_i)-u_i(\delta)}{\epsilon}
    \\
    =&\; \lim_{\epsilon \to 0} \left[v_{i,x}\cdot \frac{\log(\delta(x)+\epsilon)-\log(\delta(x))}{\epsilon}+v_{i,y}\cdot \frac{\log(\delta(y)-\epsilon)-\log(\delta(y))}{\epsilon}\right]
    \\
    =&\; \frac{v_{i,x}}{\delta(x)}-\frac{v_{i,y}}{\delta(y)}>0.
\end{align*}
Therefore, $\delta$ is not an equilibrium distribution for Cobb-Douglas utilities.

$\Leftarrow$:
Given a vector of individual distributions $(\delta_i)_{i \in N}$, suppose that each agent $i$ only contributes to charities in $T_{\delta,i}$.
In every other strategy of agent $i$, she must contribute $\epsilon$ less to at least one such charity $y\in T_{\delta,i}$ and $\epsilon'$ more to another charity $x$. 
Similar to the other direction, it can be shown that $v_{i,y}\cdot(\log(\delta(y))-\log(\delta(y)-\min\{\epsilon,\epsilon'\}))>v_{i,x}\cdot(\log(\delta(x)+\min\{\epsilon,\epsilon'\})-\log(\delta(x)))$. 
Therefore, the utility of agent $i$ is smaller than $u_i(\delta)$ and the deviation is not beneficial.
\end{proof}

As a consequence, existence and uniqueness of the equilibrium distribution carry over to Cobb-Douglas utility functions. However, other axiomatic properties from \Cref{sec:equilibrium-rule} break down. For example, the equilibrium distribution does not need to be efficient for Cobb-Douglas utilities.

\begin{example}\label{rem:cobbdouglas}
The equilibrium distribution can violate efficiency for Cobb-Douglas utilities. Suppose there are three charities and two agents with the following valuations and contributions.
\[
\begin{array}{cccc@{\hskip 2em}c}
    \toprule
          & a & b & c & C_i \\ 
         \midrule
         v_1       & 1 & 1 & 0 & 6\\
         v_2       & 0 & 1 & 1 & 6\\
        \bottomrule
\end{array}
\]
Then, the equilibrium distribution $\delta^*=(4,4,4)$ results in utilities $u_1(\delta^*)=u_2(\delta^*)=4\cdot 4=16$. However, the distribution $\delta=(3,6,3)$ provides higher Cobb-Douglas utility to both agents: $u_1(\delta)=u_2(\delta)=3\cdot 6=18$.

Furthermore, the mechanism that returns the equilibrium distribution can be manipulated. 
If the first agent only reports a positive value for charity $a$, her utility increases from $4\cdot 4=16$ to $6\cdot 3=18$.
\end{example}

An interesting question to ask is which other utility functions have the same equilibria as Leontief utilities. It turns out that, when all valuations are binary ($v_{i,x} \in \{0,1\}$), \Cref{prop:sepaddutils} also holds when the logarithms in the additive representation of Cobb-Douglas utility functions are replaced with an arbitrary strictly concave, strictly increasing function.

\begin{proposition}\label{prop:otherutilsthanCD}
    Let $h \colon \mathbb{R}_{\ge 0} \to \mathbb{R}$ be a strictly concave and strictly increasing function. Assume agent $i$'s utility function is given by $u_i(\delta)=\sum_{x \in A}v_{i,x}\cdot h(\delta(x))$ with binary valuations $v_{i,x} \in \{0,1\}$. Then, a vector of individual distributions is in equilibrium if and only if it is in equilibrium for the corresponding Leontief utilities.
\end{proposition}

\begin{proof}
We show that \Cref{lem:eq-iff-critical} applies.
Let $(\delta_i)_{i \in N}$ be a vector of individual distributions. It is straightforward to see that in equilibrium, $v_{i,x}=0$ implies $\delta_i(x)=0$ for every pair of agents and charities. Moreover, for any agent $i$ and two charities $x$ and $y$ with $v_{i,x}=v_{i,y}=1$ and $\delta(x)>\delta(y)$, agent $i$ would like to move $\epsilon=(\delta(x)-\delta(y))/2$ of contribution from $x$ to $y$, as $h(\delta(x)-\epsilon)+h(\delta(y)+\epsilon) > h(\delta(x))+h(\delta(y))$ by strict concavity of $h$. Hence, $(\delta_i)_{i \in N}$ is in equilibrium if and only if each agent contributes only to approved charities that receive the minimal total contribution under $\delta$. This set of charities coincides with $T_{\delta,i}$.
\end{proof}

In summary, only some of our results for Leontief utilities extend to other utility models. While axiomatic properties might break down (see, e.g., \Cref{rem:cobbdouglas}), equilibria and best response dynamics coincide (see  \Cref{rem: sepadd-dynamics}). In particular, the unique equilibrium distribution for Cobb-Douglas utility functions can still be computed using linear programming in the case of binary valuations.

\subsection{Generalizing Leontief utilities}

The remarkable axiomatic properties of the equilibrium distribution generalize beyond Leontief utilities. It turns out that the main ingredient for some of our results is that preferences are representable as a minimum over a set of strictly increasing functions. More specifically, assume that agent $i$'s utility function is given by
\begin{align*}
    u_i(\delta)=\min_{x \in A_i}f_{i,x}(\delta(x))
\end{align*}
where $A_i \subseteq A$ is the set of charities that ``matter'' for agent $i$ (for Leontief utilities, $v_{i,x}>0$) and $f_{i,x} \colon \mathbb{\mathbb{R}}_{\ge 0} \to \mathbb{R}$ are continuous and strictly increasing functions (for Leontief utilities, $f_{i,x}(\delta(x))=\delta(x)/v_{i,x}$).
These utility functions are still quasi-concave, as for two distributions $\delta,\delta'$ and arbitrary $\lambda \in [0,1]$,
\begin{align*}
    u_i(\lambda\cdot\delta+(1-\lambda)\cdot \delta')&=\min_{x \in A_i}f_{i,x}(\lambda\cdot\delta(x)+(1-\lambda)\cdot \delta'(x))\\
    &\geq \min_{x \in A_i} \min\{f_{i,x}(\delta(x)),f_{i,x}(\delta'(x))\}\\
    &=\min\{u_i(\delta),u_i(\delta')\}
\end{align*}
where the inequality follows from the fact that all $f_{i,x}$'s are strictly increasing. 
Existence of an equilibrium distribution (in the distribution game) then follows from classical results in equilibrium theory \citep[e.g.,][]{Debr52a}.

Adapting the definition of critical charities to
$\displaystyle
T_{\delta,i} \coloneqq 
\arg \min_{x\in A_i} f_{i,x}(\delta(x))
$, \Cref{lem:eq-iff-critical} still holds.
The proof of \Cref{thm:contribution-monotonicity} holds almost as is 
(where the ``Leontief quotients'' are replaced by the corresponding $f_{i,x}$'s); therefore, the ensuing \Cref{rem: alternative-uniqueness-proof} implies uniqueness.
In the same manner, the proof of \Cref{thm:strategyproof} can be adapted to show group-strategyproofness in this more general setting. 

On the negative side, apart from losing the connection to Nash welfare, allowing for these general utility functions might introduce asymmetries among charities. It is also unclear how agents can report their utility functions effectively and accurately to the mechanism infrastructure.

\section{Proofs omitted from \Cref{sec:equilibrium-distribution}}\label{app:equilibriumsection}

\eqiffcritical*

\begin{proof}
Let us write for brevity $f_{i,x}(t) := t / v_{i,x}$,
for all $i\in N$ and $x\in A$, 
so that $u_i(x) = \min_{x\in A_i} f_{i,x}(\delta(x))$,
and note that $f_{i,x}$ is a strictly increasing function. 

$\Rightarrow$: 
Suppose that some agent~$i$ contributes to a charity $y\not\in T_{\delta,i}$. 
So either $y\not\in A_i$ or
$\delta(y) > f_{i,y}^{-1}(u_i(\delta))$.
In both cases,
agent $i$ can reduce a small amount from $\delta_i(y)$ and distribute it equally among all charities in $T_{\delta,i}$. As 
all functions $f_{i,x}$ are strictly increasing, 
 this move strictly increases the utility of $i$. Therefore, $\delta$ is not an equilibrium distribution.

$\Leftarrow$:
Suppose each agent $i$ only contributes to charities in $T_{\delta,i}$.
In every other strategy of agent $i$, she must contribute less money to at least one such charity, $y\in T_{\delta,i}$. 
By definition of a critical charity, 
the original allocation to charity $y$ was exactly $f_{i,y}^{-1}(u_i(\delta))$, so 
the new allocation to $y$ is smaller than
$f_{i,y}^{-1}(u_i(\delta))$.
As $f_{i,y}$ is strictly increasing, 
the new utility of agent $i$ is smaller than $u_i(\delta)$ and the deviation is not beneficial.
\end{proof}

\efficientiffcritical*
\begin{proof}
As in the proof of \Cref{lem:eq-iff-critical}, we write $f_{i,x}(t) := t / v_{i,x}$.

$\Rightarrow$: Suppose that some charity $x \in \supp(\delta)$ is not critical for any agent. 
This means that, 
for each agent $i\in N$,
either $x\not\in A_i$
or $\delta(x) > f_{i,x}^{-1}(u_i(\delta))$.
Denote
\begin{align*}
D \coloneqq \delta(x) - \max_{i\in N,\, x\in A_i} \left(f_{i,x}^{-1}(u_i(\delta))\right)
\end{align*}
where our assumptions imply that $D>0$.
Construct a new distribution $\delta'$ by removing $D/2$ from charity $x$ and distributing it equally among all other charities.
We claim that $u_i(\delta')>u_i(\delta)$ for every agent $i\in N$.
Indeed, if $x\not\in A_i$, then $u_i$ does not decrease by the removal from $\delta(x)$, and strictly increases by the addition to all other charities. Otherwise,
\begin{align*}
u_i(\delta') = \min\left(f_{i,x}(\delta'(x)), \min_{y\in A_i\setminus x}
f_{i,y}(\delta'(y))
\right).
\end{align*}
Both terms are larger than $u_i(\delta)$:
\begin{itemize}
\item The former term is 
\begin{align*}
f_{i,x}(\delta(x)-D/2) 
&> f_{i,x}(\delta(x)-D) \\
&= f_{i,x}\left(\max_{j\in N,\,x\in A_j}[f_{j,x}^{-1}(u_j(\delta))]\right) \geq f_{i,x}(f_{i,x}^{-1}(u_i(\delta))) = u_i(\delta)
\end{align*}
by construction, as the $f_{i,x}$ are strictly increasing.

\item For the latter term, 
the fact that $u_i(\delta)<f_{i,x}(\delta(x))$ implies that 
$u_i(\delta)=\min_{y\in A_i\setminus x}\left(f_{i,y}(\delta(y))\right)$, and 
$\min_{y\in A_i\setminus x}\left(f_{i,y}(\delta'(y))\right)$ is strictly larger than that since each charity $y\in A\setminus x$ receives additional funding in $\delta'$.
\end{itemize}
Hence, $\delta$ is not efficient.

$\Leftarrow$: Suppose that every charity $x\in \supp(\delta)$ is critical for some agent.
Let $\delta'$ be any collective distribution different than $\delta$. Since the sum of both distributions is the same ($C_N$), there exists a charity $y\in \supp(\delta)$ with $\delta'(y)<\delta(y)$.
Let $i_y\in N$ be an agent for whom $y$ is critical in $\delta$. Then
the utility of $i_y$ is strictly smaller in $\delta'$:
\begin{align*}
u_{i_y}(\delta') 
&\leq f_{i_y,y}(\delta'(y)) && \text{(by definition of Leontief utilities, as $y\in A_{i_y}$)}
\\
&< f_{i_y,y}(\delta(y))  && \text{($\delta'(y)<\delta(y)$ by definition of $y$, and $f_{i_y,y}$ is strictly increasing)}
\\
&= u_{i_y}(\delta)  && \text{(since 
$y$ is critical for $i_y$ in $\delta$)
}
\end{align*}
so $\delta'$ does not dominate $\delta$.
Hence, $\delta$ is efficient.
\end{proof}

\efficientunique*

\begin{proof}
By Lemma \ref{lem:efficient-iff-critical}, for each $x \in \supp(\delta)$, there is an agent for whom $x$ is critical. Denote one such agent by $i_x$. Then,
\begin{align*}
\delta(x)&= v_{i_x,x}\cdot u_{i_x}(\delta) 
&& \text{(by \eqref{eq:critical-implications}, since 
$x$ is critical for $i_x$)}
\\
&=v_{i_x,x}\cdot u_{i_x}(\delta') 
&& \text{(by the lemma assumption)}
\\
&\leq \delta'(x)
 && \text{(by definition of Leontief utilities)}.
\end{align*}
The same inequality $\delta(x)\leq \delta'(x)$ trivially holds also for all $x\not \in\supp(\delta)$.
Since both distributions sum up to $C_N$, this implies $\delta=\delta'$.
\end{proof}

\subsection{Welfare-maximizing distributions}
\label{sub:fwelfare}

Let $g$ be a strictly increasing function. 
The \emph{$g$-welfare} of a distribution $\delta$ is defined as the following weighted sum:
\begin{align*}
\gwelfare(\delta) \coloneqq \sum_{i\in N} C_i\cdot g( u_i(\delta)).
\end{align*}
Quantifying welfare enables us to compare and rank all possible utility vectors, which induces 
a social welfare ordering over all distributions $\delta \in \Delta(C_N)$ by $\gwelfare(\delta)$.

Inversely, every continuous social welfare ordering without any ``welfare dependencies'' between the agents' utilities can be represented by a \gwelfare function; see Chapter~2 in the book by \citet{Moul88a} for a detailed discussion. Additionally weighting agents by their contributions, we arrive at the very expressive class of \gwelfare functions.

A distribution is called \emph{\gwelfare-maximizing} if it maximizes 
the \gwelfare, i.e., it always chooses a maximal element of the corresponding social welfare ordering.
Clearly, every \gwelfare-maximizing distribution is efficient.
When $g$ is concave (equivalently: when the induced social welfare ordering satisfies the Pigou-Dalton principle),
 a \gwelfare-maximizing distribution can be found by solving a convex program where the variables are $(u_i)_{i\in N}$ and $(\delta_x)_{x\in A}$:
 
\begin{align*}
\text{maximize} &&\sum_{i\in N} C_i\cdot g(u_i) \notag
\\\text{subject to}
&&\sum_{x\in A} \delta_x &\leq C_N
\\
 &&u_i&\leq \delta_x /v_{i,x} && \text{~for all~} i\in N,\, x\in A_i
\\
 &&u_i&\geq 0&& \text{~for all~} i\in N
 \\
&&\delta_x&\geq 0&& \text{~for all~} x\in A.
\end{align*}

The following technical lemmas prove uniqueness of the welfare-maximizing distribution when $g$ is additionally strictly concave.
\begin{lemma}
\label{lem:mixing-strictly-concave}
For every strictly concave, increasing function $g$, every constant $t\in(0,1)$, and every two distributions $\delta\neq \delta'$,
\begin{align*}
\gwelfare\left(t \delta+(1-t)\delta'\right) ~>~ \min(\gwelfare(\delta'), \gwelfare(\delta)).
\end{align*}
\end{lemma}

\begin{proof}
For every agent $i\in N$, by the concavity of the minimum operator,
\begin{align*}
u_i\left(t \delta+(1-t)\delta' \right)
\geq
t u_i(\delta)+(1-t)u_i(\delta).
\end{align*}
Therefore,
\begin{align*}
\gwelfare\left(t \delta+(1-t)\delta'\right) 
&\geq
\sum_{i \in N}C_i \cdot  g\left(t\cdot u_i(\delta)+(1-t)\cdot u_i(\delta')\right) 
\\
&>
t\sum_{i \in N}C_i \cdot  g(u_i(\delta))+(1-t)\sum_{i \in N}C_i \cdot  g(u_i(\delta')) 
\\
&=
t\cdot \gwelfare(\delta) + (1-t)\cdot \gwelfare(\delta')
\\
&\geq 
\min(\gwelfare(\delta'), \gwelfare(\delta))
\end{align*}
where the first inequality follows from monotonicity and the second one from strict concavity.
\end{proof}

\begin{lemma}
\label{lem:fmaximizing-is-unique}
For every strictly concave, increasing function $g$, there is a unique \gwelfare-maximizing distribution.
\end{lemma}
\begin{proof}
Assume for contradiction that there exist two different \gwelfare-maximizing distributions $\delta$ and $\delta'$. Since both distributions are efficient, by Lemma \ref{lem:efficientunique} they induce two different utility vectors $(u_i(\delta))_{i \in N}$ and $(u_i(\delta'))_{i \in N}$. 
By Lemma \ref{lem:mixing-strictly-concave}, for any $t\in(0,1)$,
\begin{align*}
\gwelfare\left(t \delta+(1-t)\delta'\right) &> \min(\gwelfare(\delta'), \gwelfare(\delta))
\\
&=
\gwelfare(\delta') = \gwelfare(\delta).
\end{align*}
This contradicts the assumption that $\delta$ and $\delta'$ are \gwelfare-maximizing. 
\end{proof}

\subsection{Proof of \Cref{thm:unique-eq-distribution}}\label{app:proofuniqueeq}

\uniqueeqdistribution*

We show equivalence of the equilibrium distribution and the Nash-optimal distribution by proving the two directions in \Cref{lem:nashiseq,lem:eqisnash}, respectively.

\begin{restatable}{lemma}{nashiseq}
\label{lem:nashiseq}
Every distribution that maximizes Nash welfare is an equilibrium distribution.
\end{restatable}

One way to prove \Cref{lem:nashiseq} is to analyze the KKT conditions of the constrained maximization problem corresponding to maximizing Nash welfare. 
Below, we give a more intuitive proof, which helps to illustrate the ``social'' aspect of the equilibrium distribution.
We first show that, in any other distribution, there is a set of agents who ``waste'' some of their contribution on charities that are only critical for other agents.

\begin{lemma}
\label{lem:not-eq}
If $\delta$ is an efficient distribution that is not an equilibrium distribution,
then $N$ can be partitioned into two disjoint nonempty groups of agents, $N_+$ and $N_- = N\setminus N_+$, such that
\begin{align}
\label{eq: N-}
\delta(T_{\delta, N_-}) < C_{N_-};
\\
\label{eq: N+}
\delta(T_{\delta, N_+} \setminus T_{\delta, N_-}) > C_{N_+}.
\end{align}
\end{lemma}
\begin{proof}
Let $(\delta_i)_{i\in N}$ be any decomposition of $\delta$. 
Construct a directed graph $G$ in which
the nodes correspond to agents, and there is an arc $i\to j$ if and only if $\delta_i(T_{\delta,j})>0$, that is, agent $i$ contributes to a critical charity of $j$.
We call the arc $i\to j$ \emph{strong} if  $\delta_i(T_{\delta,j}\setminus T_{\delta,i})>0$, that is, agent $i$ contributes to a charity that is critical for $j$ but not for~$i$. Otherwise, we call the arc $i\to j$  \emph{weak}.
Since $\delta$ is not an equilibrium distribution, by \Cref{lem:eq-iff-critical}, there is an agent, say Agent 1,  who contributes to a charity $x\not\in T_{\delta,1}$. 
Since $\delta$ is efficient, by \Cref{lem:efficient-iff-critical}, $x$ is critical to some other agent, say Agent 2, so $G$ contains a strong arc $1\to 2$.

If the strong arc is part of a directed cycle, then we can move a sufficiently small amount $\epsilon$  along the cycle without changing $\delta$. 
In detail, suppose without loss of generality that the cycle is $1\to 2\to\cdots\to k \to 1$, 
where the involved charities are 
$
x_1\in T_{\delta,1},~
x_2\in T_{\delta,2}\setminus T_{\delta,1},~
x_3\in T_{\delta,3},~
x_4\in T_{\delta,4},~
\dots,~
x_k\in T_{\delta,k}.$
We assume that $x_2$ 
is in $T_{\delta,2}\setminus T_{\delta,1}$ since the arc $1\to 2$ is strong; in particular, $x_2$ 
must be different than~$x_1$. 
The other arcs may be strong or weak, and some of the $x_i$ may coincide.
For every $i\in\{1,\ldots,k-1\}$,
move a small amount $\epsilon>0$ 
from $\delta_i(x_{i+1})$
to $\delta_i(x_i)$;
move the same $\epsilon$ from 
$\delta_k(x_{1})$ to $\delta_k(x_{k})$.
Note that the decomposition changes, but the total $\delta$ remains the same.
Increase~$\epsilon$ until one arc of the cycle disappears or the strong arc becomes weak; note that we may need to change the involved charities in the cycle during the process.
Repeat this cycle-removal procedure until all strong arcs are not part of any directed cycle.
This process is guaranteed to terminate since in each cycle removal, either the respective strong arc becomes weak, or the cycle it is part of is removed. Furthermore, no new (strong) arcs are created as agents do not contribute to additional charities, and the overall distribution $\delta$ together with the set of critical charities does not change.

Let $G$ be the graph of the resulting decomposition. Since the total distribution is still~$\delta$, which is efficient but not an equilibrium distribution, $G$ still has at least one strong arc, say $j\to k$.
Let $N_+$ be the set of agents accessible from $k$ via a directed path (where $k\in N_+$), and let $N_-\coloneqq N\setminus N_+$.
Since $j\to k$ is not part of any directed cycle,
$j\in N_-$.

Due to the strong arc $j\to k$,
agents of $N_-$ waste some of their own contributions on critical charities of $N_+$ that are not critical for themselves. 
Moreover, the critical charities of $N_-$ do not receive any donations from agents of $N_+$,
since they are not accessible from $N_+$.
This proves \eqref{eq: N-}.

In contrast, the agents in $N_+$ spend all their contributions on their own critical charities that are not critical charities of agents outside $N_+$. In addition, they receive some donations from agents of $N_-$. This proves \eqref{eq: N+}.
\end{proof}

\begin{proof}[Proof of \Cref{lem:nashiseq}]
Let $\delta$ be an efficient distribution that is not an equilibrium distribution. We prove that $\delta$ is not Nash-optimal.

Let $N_-$ and $N_+$ be the subsets of agents defined in \Cref{lem:not-eq}.
If $\delta(T_{\delta, N_-}) = 0$, then $\Nash(\delta) = -\infty$ and $\delta$ is clearly not Nash-optimal, so we may assume that $\delta(T_{\delta, N_-}) > 0$.
We construct a new distribution $\delta'$ in the following way.
\begin{itemize}
\item Remove a small amount $\epsilon$ from $\delta(T_{\delta, N_+} \setminus T_{\delta, N_-})$, such that each charity loses proportionally to its current distribution. That is, for each charity $x\in T_{\delta, N_+} \setminus T_{\delta, N_-}$,
the new distribution is $\delta'(x) \coloneqq \delta(x)\cdot [1-\epsilon/\delta(T_{\delta, N_+} \setminus T_{\delta, N_-})]$.
\item Add this $\epsilon$ to $\delta(T_{\delta, N_-})$ such that each charity gains proportionally to its current distribution. 
That is, for each charity $y\in T_{\delta, N_-}$,
the new distribution is $\delta'(y) \coloneqq \delta(y)\cdot [1+\epsilon/\delta(T_{\delta, N_-})]$.
\end{itemize}
Choose $\epsilon$ sufficiently small such that the sets of critical charities of agents in $N_-$ do not change (that is, no new charities become critical for them).
This redistribution has the following effect on the agents' utilities:
\begin{itemize}
\item The utility of each agent $i\in N_+$ may decrease by a factor of up to $[1-\epsilon/\delta(T_{\delta, N_+} \setminus T_{\delta, N_-})]$. Therefore, the contribution to Nash welfare may decrease by at most
$
\Delta_{N_+}(\epsilon) \coloneqq 
C_{N_+}\cdot \log[1-\epsilon/\delta(T_{\delta, N_+} \setminus T_{\delta, N_-})].
$
We have $\lim_{\epsilon\to 0}\Delta_{N_+}(\epsilon)/\epsilon = -C_{N_+}/\delta(T_{\delta, N_+} \setminus T_{\delta, N_-})$, which is larger than $-1$ by inequality \eqref{eq: N+}.

\item The utility of each agent $i\in N_-$  increases by a factor of $[1+\epsilon/\delta(T_{\delta, N_-})]$. 
Therefore, the contribution to Nash welfare increases by 
$
\Delta_{N_-}(\epsilon) 
\coloneqq
C_{N_-}\cdot \log[1+\epsilon/\delta(T_{\delta, N_-})].
$
We have $\lim_{\epsilon\to 0}\Delta_{N_-}(\epsilon)/\epsilon = C_{N_-}/\delta(T_{\delta, N_-})$, which is larger than $1$ by inequality~\eqref{eq: N-}.
\end{itemize}
The overall difference in Nash welfare is $\Delta(\epsilon) \coloneqq \Delta_{N_+}(\epsilon)+\Delta_{N_-}(\epsilon)$, and we have 
$\lim_{\epsilon\to 0}\Delta(\epsilon)/\epsilon > -1+1 = 0$, so $\Delta(\epsilon)>0$ for sufficiently small $\epsilon$.
Therefore, $\Nash(\delta')> 
\Nash(\delta)$, so $\delta$ was not Nash-optimal, 
completing the proof.
\end{proof}

For the other direction, the following lemma proves that every equilibrium distribution is Nash-optimal.

\begin{restatable}{lemma}{eqisnash}
\label{lem:eqisnash}
Every equilibrium distribution maximizes Nash welfare.
\end{restatable}

\begin{proof}
Let $\delta^*$ be an equilibrium distribution.
For any distribution $\delta$, we derive an upper bound for $\Nash(\delta)$ in terms of $\delta^*$.
We show that this upper bound is maximized when $\delta=\delta^*$ 
and is equal to $\Nash(\delta)$ for $\delta=\delta^*$.
Thus, $\Nash(\delta)\leq \Nash(\delta^*)$ so $\delta^*$ maximizes the Nash welfare.

Formally, 
let $(\delta^*_i)_{i\in N}$ be any Nash-equilibrium decomposition of $\delta^*$ (satisfying \Cref{lem:eq-iff-critical}), and let $N_{\delta^*,x} \coloneqq \{i \colon x\in T_{\delta^*,i}\}$ be the set of agents for whom $x$ is critical in $\delta^*$.
For every distribution $\delta$ with $\Nash(\delta) > -\infty$, we have
\begin{align*}
\Nash(\delta)&=\sum_{i\in N} C_i \log(u_i(\delta))
\\
&= 
\sum_{i\in N}
\left(
\sum_{x \in T_{\delta^*,i}} \delta^*_i(x)
\right)
\log(u_i(\delta))
&&\text{(by \eqref{eq:dec-ci-critical})}
\\
&\leq 
\sum_{i\in N} \sum_{x \in T_{\delta^*,i}} \delta^*_i(x) \cdot \log\left(\frac{\delta(x)}{v_{i,x}}\right) 
\\
&=
\sum_{x\in A} 
\left(\sum_{i \in N_{\delta^*,x}} \delta^*_i(x)
 \cdot \log\left(\frac{\delta(x)}{v_{i,x}}\right)\right)
\\
&=
\sum_{x\in A} 
\left(\sum_{i \in N_{\delta^*,x}} \delta^*_i(x)
\right) \log(\delta(x))
- \sum_{x\in A} 
\left(\sum_{i \in N_{\delta^*,x}} \delta^*_i(x)
\log(v_{i,x})\right) 
\\
&=
\sum_{x\in A} 
\delta^*(x) \log(\delta(x))
- \sum_{x\in A} 
\left(\sum_{i \in N_{\delta^*,x}} \delta^*_i(x)
\log(v_{i,x})\right)\text.
&&\text{(by \eqref{eq:dec-dx})}
\end{align*}
We claim that, for every fixed $\delta^*$, the latter expression is 
maximized for $\delta=\delta^*$.
This follows from Gibbs' inequality; we provide the proof here for completeness.
Note that the second term is independent of $\delta$.
As for the first term
$\sum_{x\in A} 
\delta^*(x) \log(\delta(x))$,
consider the optimization problem of maximizing 
$\sum_{x\in A} 
\delta^*(x) \log(\delta(x))$ 
subject to 
$\sum_{x \in A}\delta(x) = \sum_{x \in A}\delta^*(x)$ (note that $\delta^*$ is a constant in this problem).
Its Lagrangian is 
\begin{align*}
\sum_{x\in A} 
\delta^*(x) \log(\delta(x))
+
\lambda\cdot\left(
\sum_{x \in A}\delta^*(x) - 
\sum_{x \in A}\delta(x) 
\right).
\end{align*}
Setting the derivative with respect to $\delta(x)$ to $0$ gives
$\delta^*(x)/\delta(x) = \lambda$ for all $x\in A$. Since $\sum_{x \in A}\delta(x) = \sum_{x \in A}\delta^*(x)$, we must have $\lambda=1$, so $\delta=\delta^*$.
This means that
\begin{align*}
\Nash(\delta)\leq 
\sum_{x\in A} 
\delta^*(x) \log(\delta^*(x))
-
\text{const}(\delta^*).
\end{align*}
For $\Nash(\delta^*)$, the same derivation holds, but the inequality becomes an equality, since in equilibrium, 
$\delta^*_i(x)>0$ only if 
$u_i(\delta^*) = \delta^*(x)/v_{i,x}$.
Therefore, 
\begin{align*}
\Nash(\delta)\leq 
\Nash(\delta^*),
\end{align*}
so $\delta^*$ is Nash-optimal.
\end{proof}

Since the logarithm function is strictly concave,
\Cref{lem:fmaximizing-is-unique} implies that there is a unique distribution that maximizes Nash welfare.
Hence, \Cref{lem:nashiseq,lem:eqisnash} entail that there is a unique equilibrium distribution.

\subsection{Proof of \Cref{prop:Lindahleq}}\label{app:lindahl}

\Lindahleq*

\begin{proof}
We first show that every equilibrium distribution is a Lindahl equilibrium.
Let $\delta^*$ be an equilibrium distribution and $(\delta^*_i)_{i\in N}$ a corresponding Nash equilibrium.

Define personalized price functions 
\[
  p_{i}(x) \coloneqq 
  \begin{cases}
  \frac{\delta^*_i(x)}{\delta^*(x)} &\text{if $\delta^*(x)>0$,} \\
  0 & \text{if $\delta^*(x)=0$.}
  \end{cases}
\]
Note that $\delta^*(x)>0$ implies that charity $x$ is valuable to at least one agent.

We now show that, with these price functions, $\delta^*$ satisfies both conditions in the Lindahl equilibrium definition.

\emph{Condition 1}:
For each $i\in N$, 
\begin{align*}
 &\; \sum_{x\in A} p_i(x)\delta^*(x)  
 \\
 =&\;
 \sum_{x\in \supp(\delta^*)} p_i(x)\delta^*(x)  
 +  \sum_{x\not\in \supp(\delta^*)} p_i(x)\delta^*(x)   
 \\
 =&\;
 \sum_{x\in \supp(\delta^*)} \delta^*_i(x) + 0
 \\
 =&\;
 \sum_{x\in A} \delta^*_i(x)  && \text{(as $x\not\in \supp(\delta^*)$ implies $\delta^*_i(x)=0$)}
\\
 =&\; C_i.
 \end{align*}
To prove maximality, 
let $y \in \mathbb{R}_{\geq 0}^A$ be any other vector with 
$\sum_{x\in A} p_i(x) y(x) = C_i$.
By \Cref{lem:eq-iff-critical}, $p_{i}(x)=0$ for every project $x$ not critical for $i$. Therefore, 
\begin{align*}
\sum_{x\in T_{\delta^*,i}} p_i(x) y(x) = C_i = \sum_{x\in T_{\delta^*,i}} p_i(x) \delta^*(x).    
\end{align*}
There must be at least one $x\in T_{\delta^*,i}$ with $\delta^*_i(x)>0$, which implies $p_i(x)>0$. Hence,
there must be at least one $x\in T_{\delta^*,i}$ with 
$y(x)\leq \delta^*(x)$.
By \Cref{lem:twodistributions}, 
$u_i(y) \leq u_i(\delta^*)$.

\emph{Condition 2}:
By our definition of prices, $\sum_{i \in N}p_i(x)=\sum_{i \in N}\delta^*_i/\delta^*(x) = \delta^*(x)/\delta^*(x)= 1$ for all $x \in \supp(\delta^*)$ and $\sum_{i \in N}p_i(x)=0$ for all $x \not \in \supp(\delta^*)$.
Hence, $\delta^*$ is a Lindahl equilibrium as claimed.

 Next, we show that every Lindahl equilibrium is an equilibrium distribution.
Let $\delta$ be a Lindahl equilibrium and let $p_1,\ldots,p_n\in \mathbb{R}_{\geq 0}^A$ be corresponding personalized price functions satisfying Conditions 1 and 2.

We first prove that $\delta\in \Delta(C_N)$. We have
\begin{align*}
    \sum_{x\in A} \delta(x) &= \sum_{x\in \supp(\delta)} \delta(x)
    \\
    &=
    \sum_{x\in \supp(\delta)} \left(\sum_{i\in N}p_i(x)\right) \delta(x)
     && \text{(by Condition 2)}
    \\
    &=
    \sum_{i\in N} \sum_{x\in \supp(\delta)} p_i(x) \delta(x)
    \\
    &=
    \sum_{i\in N} C_i 
         && \text{(by Condition 1)}
    \\
    &=
    C_N.
\end{align*}

Note that $p_{i}(x)=0$ for $x \in \supp(\delta)\setminus T_{\delta,i}$; otherwise, as $\delta(x)>0$, agent $i$ spends a positive amount $p_i(x)\delta(x)$ on $x$, and could transfer parts of it to her critical charities which improves her utility, in contradiction to Condition 1 of Lindahl equilibrium.

We claim that for any group of agents $N_- \subseteq N$, it holds that $C_{N_-} \leq \delta\left(T_{\delta,N_-}\right)$. 
To see this, note that
\begin{align*}
C_{N_-}&=\sum_{i \in N_-}C_i=\sum_{i \in N_-} \sum_{x \in A}p_{i}(x) \delta(x)
&&\text{(as $C_i = \sum_{x\in A} p_i(x) \delta(x)$ by Condition 1)}
\\
&=\sum_{i \in N_-} \sum_{x \in T_{\delta,i}}p_{i}(x) \delta(x)
&&\text{(as either $\delta(x)=0$ or $p_{i}(x)=0$ for $x \not\in T_{\delta,i}$)}
\\
&=\sum_{x \in T_{\delta,N_-}}\sum_{i \in N_-}p_{i}(x) \delta(x)
\\ 
&\leq \sum_{x \in T_{\delta,N_-}}\delta(x)
&&\text{(as $\sum_{i \in N_-}p_{i}(x) \leq 1$ for all $x \in A$ by Condition 2)}
\\
&=\delta\left(T_{\delta,N_-}\right).
\end{align*}
By (\ref{eq: N-}) in \Cref{lem:not-eq}, $\delta$ is the equilibrium distribution.
\end{proof}

\section{Proofs omitted from \Cref{sec:equilibrium-rule}}
\label{app:axioms}

\subsection{Strategyproofness and participation}
We prove that \edr is group-strategyproof by leveraging the following lemma.

\begin{restatable}{lemma}{twodistributions}
\label{lem:twodistributions}
Let $\delta^1$ and $\delta^2$ be two distributions, and $i\in N$ an agent.

(a) If $u_i(\delta^2) \geq u_i(\delta^1)$, 
then every charity in $T_{\delta^1,i}$ receives at least as much funding in $\delta^2$, that is:
$\delta^2(y) \geq \delta^1(y)$ for all $y\in T_{\delta^1,i}$.

(b) Similarly, if $u_i(\delta^2) > u_i(\delta^1)$, 
then $\delta^2(y) > \delta^1(y)$ for all $y\in T_{\delta^1,i}$.
\end{restatable}

\begin{proof}
As in the proof of \Cref{lem:eq-iff-critical}, we write $f_{i,x}(t) := t / v_{i,x}$ for all $x\in A_i$.

For (a),
for every charity $y \in T_{\delta^1,i}$,
we have 
\begin{align*}
\delta^1(y)&=f_{i,y}^{-1}(u_i(\delta^1)) && \text{(as $y$ is critical for $i$ in $\delta^1$)}
\\
&\leq f_{i,y}^{-1}(u_i(\delta^{2})) && \text{(by assumption $u_i(\delta^1)\leq u_i(\delta^2)$, and $f_{i,y}$ is increasing)}
\\
&=f_{i,y}^{-1}\left(\min_{x \in A_i}f_{i,x}(\delta^{2}(x))\right) && \text{(by definition of Leontief utilities)}
\\
& \leq f_{i,y}^{-1}(f_{i,y}(\delta^{2}(y))) && \text{(since $y \in T_{\delta^1,i}\subseteq A_i$)}
\\
& = \delta^{2}(y).
\end{align*}
For (b), the first inequality 
becomes strict.
\end{proof}

\strategyproof*

\begin{proof}
Given a profile $P$, we define a \emph{manipulation} for group $G$ as a profile $P'\neq P$ in which $v_{i,x}=v_{i,x}'$ for all $i\in N\setminus G$.
A manipulation is \emph{profitable} if $u_j(f(P'))\geq u_j(f(P))$ for all $j\in G$, and $u_i(f(P'))> u_i(f(P))$ for at least one $i\in G$.
Group-strategyproofness is equivalent to the lack of profitable manipulations.

Suppose by contradiction that some group of agents has a profitable manipulation. Let $G \subseteq N$ be an inclusion-maximal such group. Denote by  $\delta^P$ and $\delta^{P'}$ the outcomes of \edr before and after the manipulation.
Since the manipulation is profitable, $u_j(\delta^{P'}) \ge u_j(\delta^P)$ for all $j \in G$ and $u_i(\delta^{P'}) > u_i(\delta^P)$ for at least one $i \in G$.
By \Cref{lem:twodistributions}, 
$\delta^{P'}(x) \geq \delta^P(x)$ for every charity $x$ that belongs to $T_{\delta^P,j}$ for some $j\in G$, and $\delta^{P'}(x) > \delta^P(x)$ for every charity $x$ in $T_{\delta^P,i}$.
This implies 
\begin{align}
\label{eq:strategyproofness}
\delta^{P'}\left(\bigcup_{j \in G}T_{\delta^P,j}\right)>\delta^{P}\left(\bigcup_{j \in G}T_{\delta^P,j}\right).
\end{align}

We decompose both outcomes as Nash equilibria $(\delta^P_i)_{i \in N}$ and $(\delta^{P'}_i)_{i \in N}$, i.e., $\delta^{P} = \sum_{i \in N}\delta^{P}_i$ and $\delta^{P'} = \sum_{i \in N}\delta^{P'}_i$
satisfying \Cref{lem:eq-iff-critical}.
Since $C'_G \leq C_G$, inequality \eqref{eq:strategyproofness} above must hold for the individual distribution of at least one agent $k\in N \setminus G$, that is, 
\begin{align}
\label{eq:strategyproofness-2}
\delta_k^{P'}\left(\bigcup_{j \in G}T_{\delta^P,j}\right)>\delta_k^P\left(\bigcup_{j \in G}T_{\delta^P,j}\right).
\end{align}
Consequently, at least one charity  $x_G \in \cup_{j \in G}T_{\delta^P,j}$ has $\delta_k^{P'}(x_G)>\delta_k^P(x_G)$.
By \Cref{lem:eq-iff-critical}, $x_G$ must be critical for $k$ in $\delta^{P'}$.
Therefore,
\begin{align*}
    u_k(\delta^{P'})
    &=f_{k,x_G}(\delta^{P'}(x_G)) && \text{(as $x_G$ is critical for $k$ in $\delta^{P'}$)}
    \\
    &\geq f_{k,x_G}(\delta^P(x_G)) && \text{(by \Cref{lem:twodistributions}, as $x_G\in T_{\delta^P,j}$ for some $j\in G$)}\\
    &\geq u_k(\delta^P)  && \text{(by definition of Leontief utilities)},
\end{align*}
so agent $k$'s utility is not decreased by the group's manipulation. 
Consequently, $k$ could be added to $G$---contradicting the maximality of $G$.

We conclude that no group of agents has a profitable manipulation and thus \edr is group-strategyproof.
\end{proof}

The above proof shows also that, if the total contribution decreases ($C_G'  <C_G$), then the utility of at least one agent in $G$ has to \emph{strictly} decrease under \edr.
Otherwise, a weak inequality analogous to \eqref{eq:strategyproofness} holds,
but since $\sum_{i \in G}\delta^{P'}_i\left(\bigcup_{j \in G}T_{\delta^P,j}\right)<\sum_{i \in G}\delta^{P}_i\left(\bigcup_{j \in G}T_{\delta^P,j}\right)$,
this weak inequality still implies the strict inequality \eqref{eq:strategyproofness-2}, and the rest of the argument applies.
In particular, an agent receives \emph{strictly} more utility when she increases her contribution.

\participation*

\begin{proof}
Let $P'$ be the profile where, compared to $P$, one agent $j$ increased her contribution by $Z>0$. 
Let $\delta^{P} \in \Delta(C_N)$ and $\delta^{P'} \in \Delta(C_N+Z)$ be the respective outcomes of \edr.

We claim that $\frac{u_j(\delta^{P'})}{u_j(\delta^P)}\geq \frac{C_N+Z}{C_N}$.
To see this, 
define $\delta'=\frac{C_N+Z}{C_N} \cdot \delta^{P}$
and $\delta^{''}=\frac{C_N}{C_N+Z} \cdot \delta^{P'}$
 such that $\delta' \in \Delta(C_N+Z)$ and $\delta^{''} \in \Delta(C_N)$. 
 Denote by $\nash_P(\delta)$ the weighted product of agents' utilities in profile $P$ and distribution $\delta$ (the exponent of the Nash welfare as previously defined).
 Then,
\begin{align*}
    1
    &\leq \frac{\nash_{P'}(\delta^{P'})}{\nash_{P'}(\delta')} &&\text{(by maximality of $\delta^{P'}$ in $\Delta(C_N+Z)$)}
    \\
    &=\frac{\nash_P(\delta^{P'})}{\nash_P(\delta')}\cdot \frac{u_j(\delta^{P'})^{Z}}{u_j(\delta')^{Z}} &&\text{(as agent $j$ increased contribution by $Z$)}
    \\
    &=\left(\frac{C_N+Z}{C_N}\right)^{C_N} \cdot \frac{\nash_P(\delta^{''})}{\nash_P(\delta')} \cdot \frac{u_j(\delta^{P'})^{Z}}{u_j(\delta')^{Z}} &&\text{(as $\delta^{P'}=\frac{C_N+Z}{C_N} \cdot \delta^{''}$)}
    \\
    &=\frac{\nash_P(\delta^{''})}{\nash_P(\delta^P)} \cdot \frac{u_j(\delta^{P'})^{Z}}{u_j(\delta')^{Z}} &&\text{(as $\delta'=\frac{C_N+Z}{C_N} \cdot \delta^{P}$)}
    \\
    &\leq \frac{u_j(\delta^{P'})^{Z}}{u_j(\delta')^{Z}} &&\text{(by maximality of $\delta^P$ in $\Delta(C_N)$)}.
\end{align*}
Thus, $u_j(\delta^{P'}) \geq u_j(\delta')=\frac{C_N+Z}{C_N} \cdot u_j(\delta^P)$.
\end{proof}

\subsection{Monotonicity conditions}

\prefmon*

\begin{proof}
Let $P$ be a profile and $P'$ a modified profile where one agent $i$ increases her valuation for one charity $x$
(that is, $v'_{i,x}>v_{i,x}$ and $v'_{i,y}=v_{i,y}$ for all $y \in A\setminus x$).
Let $\delta^P$ and $\delta^{P'}$ be the respective outcomes of \edr. We need to show that $\delta^{P'}(x) \geq \delta^P(x)$.

Let $u_i$ and $u_i'$ be agent $i$'s Leontief utility functions in the two profiles. 
By definition of Leontief utilities, $u_i'(\delta^P) = \min (u_i(\delta^P), \delta^P(x)/v_{i,x}')$.
We consider two cases, depending on which of the two expressions within the minimum is larger.

\emph{Case 1}: $u_{i}(\delta^P)<\delta^P(x)/v_{i,x}'$.
Then 
$u_i'(\delta^P)=u_i(\delta^P)$, and all charities in $T_{\delta^P,i}$ remain critical for~$i$ in the new profile.
Therefore, by \Cref{lem:eq-iff-critical}, 
$\delta^P$ is still an equilibrium distribution for $P'$. By uniqueness of the equilibrium distribution, $\delta^{P'}(x)=\delta^P(x)$.

\emph{Case 2}: $u_{i}(\delta^P)\geq  \delta^P(x)/v_{i,x}'$.
By definition of Leontief utilities, 
\begin{align*}
\frac{\delta^{P'}(x)}{v'_{i,x}} \geq u_i'(\delta^{P'}).
\end{align*}
By strategyproofness (Theorem \ref{thm:strategyproof}),
\begin{align*}
u_i'(\delta^{P'})
~\geq~
u_i'(\delta^P).
\end{align*}
By definition of Leontief utilities,
\begin{align*}
u_i'(\delta^P) &= \min\left(u_i(\delta^P),~ \frac{\delta^P(x)}{v'_{i,x}}\right)
= \frac{\delta^P(x)}{v_{i,x}'},
\end{align*}
since by assumption $u_{i}(\delta^P)\geq  \delta^P(x)/v_{i,x}'$.
Combining these three inequalities yields $\delta^{P'}(x)\geq \delta^{P}(x)$, as desired.
\end{proof}

\contributionmon*

\begin{proof}
Let us write for brevity $f_{i,x}(t) := t / v_{i,x}$ for all $i\in N$ and $x\in A_i$, and note that $f_{i,x}$ is a strictly increasing function.

Let $P$ and $P'$ be two such profiles,
so that $C_i'\geq C_i$ for all $i\in N$.
Let $\delta$ and $\delta'$ be the outcomes of \edr corresponding to profiles $P$ and $P'$, respectively. 
Choose arbitrary decompositions of $\delta$ and $\delta'$ into Nash equilibria, i.e., choose vectors of individual distributions $(\delta_i)_{i \in N}$ and $(\delta'_i)_{i \in N}$ that satisfy \Cref{lem:eq-iff-critical}, such that $\delta = \sum_i \delta_i$ and $\delta' = \sum_i \delta'_i$.

Let $A^-$, $A^=$, and $A^+$ be the sets of all charities $x\in A$ with 
$\delta'(x)<\delta(x)$, $\delta'(x)=\delta(x)$, and $\delta'(x)>\delta(x)$, respectively.
Assume for contradiction that $A^-$ is not empty.
Thus, $\sum_{i \in N}\delta'_i(A^-)<\sum_{i \in N}\delta_i(A^-)$, 
so there must be an agent $i\in N$ with
$\delta'_i(A^-)<\delta_i(A^-)$,
and a charity $y \in A^-$  
with $\delta'_i(y)<\delta_i(y)$.
But $\delta'_i(A) = C'_i\geq C_i = \delta_i(A)$,
so $\delta'_i(A^= \cup A^+)>\delta_i(A^= \cup A^+)$,
and so there must be a charity $z \in A^= \cup A^+$ with $\delta'_i(z)>\delta_i(z)\ge 0$.
By  \Cref{lem:eq-iff-critical},
charities $z$ and $y$ are critical for $i$ under $\delta'$ and $\delta$, respectively.
This, in particular, implies that $v_{i,z}>0$ and $v_{i,y}>0$. Therefore, 
\begin{align*}
f_{i,z}(\delta'(z))
\le 
f_{i,y}(\delta'(y))
<
f_{i,y}(\delta(y))
\le 
f_{i,z}(\delta(z)),
\end{align*}
where the first and last inequalities follow from the definition of critical charities, and the middle inequality follows from $f_{i,y}$ being strictly increasing and $y\in A^-$.
As $f_{i,z}$ is strictly increasing,
this implies $\delta'(z)<\delta(z)$, 
a contradiction to $z \in A^= \cup A^+$.
\end{proof}

\begin{samepage}
\begin{remark}\label{rem: alternative-uniqueness-proof} 
\Cref{thm:contribution-monotonicity} yields an alternative proof of the uniqueness of equilibrium distributions, which does not rely on the equivalence with Nash welfare optimality.
If $\delta$ and $\delta'$ are equilibrium distributions for the same profile, then both $\delta'(x)\geq \delta(x)$ and 
$\delta(x)\geq \delta'(x)$ must hold for every charity $x\in A$, which implies $\delta'=\delta$.
\end{remark}
\end{samepage}

\section{Proofs omitted from \Cref{sec: br-dynamics}}
\label{app: dynamics-leontief}

\uniquebestresponse*
\begin{proof}
Since a best response corresponds to a solution of a maximization problem over the closed and bounded set of possible distributions $\delta_{i_t} + \sum_{j \neq i_t}\deltat_j$ 
with the continuous objective function $u_{i_t}$, existence is guaranteed.

To show uniqueness, observe that for the distribution in round $t+1$ (which for simplified notation we denote by $\delta\coloneqq\deltat[t+1]$), we have
$\delta_{i_t}(T_{\delta,i_t})=C_{i_t}$, that is, agent $i_t$ distributes all her contribution on her critical charities in $\delta$.
In any other response $\delta_{i_t}'$, agent $i_t$ must contribute less to at least one charity of $T_{\delta,i_t}$. Therefore, her utility must be lower than $u_{i_t}(\delta)$, so $\delta_{i_t}'$ cannot be a best response.
\end{proof}

\lempotentialpairwise*
\begin{proof}
The increase in potential is given by
\begin{align*}
    &\;\Phi(\delta_1,\ldots,\delta'_i,\ldots,\delta_n) 
    - \Phi(\delta_1,\ldots,\delta_i,\ldots,\delta_n)
    \\
    =&\;(\delta_i(x)-\epsilon)\log\left(\frac{v_{i,x}}{\delta(x)-\epsilon}\right)
    +(\delta_i(y)+\epsilon)\log\left(\frac{v_{i,y}}{\delta(y)+\epsilon}\right)
    \\
    &\;+\sum_{j \in N\setminus i:\,\delta_j(x)>0}\delta_j(x)\log\left(\frac{v_{j,x}}{\delta(x)-\epsilon}\right)
    +\sum_{j \in N\setminus i:\,\delta_j(y)>0}\delta_j(y)\log\left(\frac{v_{j,y}}{\delta(y)+\epsilon}\right)
    \\
    &\;-\delta_i(x)\log\left(\frac{v_{i,x}}{\delta(x)}\right)
    -\delta_i(y)\log\left(\frac{v_{i,y}}{\delta(y)}\right)
    \\
    &\;-\sum_{j \in N\setminus i:\,\delta_j(x)>0}\delta_j(x)\log\left(\frac{v_{j,x}}{\delta(x)}\right)
    -\sum_{j \in N\setminus i:\,\delta_j(y)>0}\delta_j(y)\log\left(\frac{v_{j,y}}{\delta(y)}\right)
    \\
    =&\;\sum_{j \in N:\,\delta_j(x)>0}\delta_j(x) \log\left(\frac{\delta(x)}{\delta(x)-\epsilon}\right)
    +\sum_{j \in N:\,\delta_j(y)>0}\delta_j(y) \log\left(\frac{\delta(y)}{\delta(y)+\epsilon}\right)
    \\
    &\;+\epsilon \left(\log \left(\frac{v_{i,y}}{\delta(y)+\epsilon}\right)-\log \left(\frac{v_{i,x}}{\delta(x)-\epsilon}\right)\right)
    \\
    =&\;\delta(x) \log\left(\frac{\delta(x)}{\delta(x)-\epsilon}\right)
    +\delta(y) \log\left(\frac{\delta(y)}{\delta(y)+\epsilon}\right)
    \\
    &\;+\epsilon \left(\log \left(\frac{v_{i,y}}{\delta(y)+\epsilon}\right)-\log \left(\frac{v_{i,x}}{\delta(x)-\epsilon}\right)\right)
    \\
    >&\;0,
\end{align*}
where the last term is nonnegative because $\frac{\delta(x)-\epsilon}{v_{i,x}} \geq \frac{\delta(y)+\epsilon}{v_{i,y}}$, and the first two terms sum up to something strictly positive which can be seen by using $\log(1+x)>\frac{x}{1+x}$ for $x> -1$ and $x \neq 0$:
\begin{align*}
    \delta(x) \log\left(\frac{\delta(x)}{\delta(x)-\epsilon}\right)
    +\delta(y) \log\left(\frac{\delta(y)}{\delta(y)+\epsilon}\right)
    &>\delta(x)\cdot \frac{\frac{\epsilon}{\delta(x)-\epsilon}}{1+\frac{\epsilon}{\delta(x)-\epsilon}}
    +\delta(y)\cdot \frac{\frac{-\epsilon}{\delta(x)+\epsilon}}{1+\frac{-\epsilon}{\delta(x)+\epsilon}}
    \\
    &=\delta(x)\cdot\frac{\epsilon}{\delta(x)}+\delta(y)\cdot\frac{-\epsilon}{\delta(y)}
    \\
    &=0. 
\end{align*}
This completes the proof.
\end{proof}

\subsection{Proof of \Cref{thm:leon-dynamics}}\label{sec:proofofthm:leondynamics}

\dynamics*
The proof will proceed in two steps: First, we will show that the amount an arbitrary
agent wants to redistribute converges to 0. Then, we will conclude that this can only be
the case if the dynamics converges to the equilibrium distribution.

The function $\Phi$ defined by \eqref{eq:potential}
is bounded on $\Delta(C_N)$, since
\begin{align*}
    \Phi(\delta_1,\ldots,\delta_n)&=\sum_{i \in N}\sum_{x \in A_i}\delta_i(x)\log \left(\frac{v_{i,x}}{\delta(x)}\right)
    \\
    &=\sum_{i \in N}\sum_{x \in A_i}\delta_i(x)\log(v_{i,x})-\sum_{x \in A}\delta(x)\log(\delta(x))
    \\
    &\le \sum_{i \in N}\sum_{x \in A_i}\delta_i(x)\log(v_{i,x})+\frac{m}{\mathrm{e}} < \infty.
\end{align*}
Therefore, the sequence $(\Phi(\deltat[t]))_{t \in \mathbb{N}}$ has to converge to some limit. We denote this limit by $\phi^*$.

We will now show that the amount an arbitrary agent wants to redistribute converges to $0$. 
By assumption, there exists a round $T \le K$ by which all agents have already appeared at least once in $\mathcal{S}$.
It is sufficient to prove the theorem for the subsequence starting at $T$. Therefore,
from now on, we assume without loss of generality that at round $t=0$, all agents have already appeared at least once in $\mathcal{S}$, 
and thus, have contributed
the entire amount $C_i$.

Denote the amount of shifted contributions in round $t$ by $c_t$:
\begin{align*}
    c_t\coloneqq\frac{1}{2}\lVert\deltat-\deltat[t+1]\rVert_1.
\end{align*} 
When moving from $\deltat$ to $\deltat[t+1]$ in round $t$, agent $i_t$ redistributes $c_t$ from a set of charities $A^-_{i_t}$ to another set $A^+_{i_t}$ with $A^+_{i_t} \cap A^-_{i_t} = \emptyset$. Since the agent is only allowed to redistribute her individual distribution, $c_t\leq \deltat_{i_t}(A^-_{i_t})$. Furthermore, since she redistributes according to her best response, she gives money only to charities that are  critical to her in the new distribution, so $\deltat[t+1]_{i_t}(x)=0$ for all $x \in A$ with $\deltat[t+1](x)/v_{i_t,x}>u_{i_t}(\deltat[t+1])$
and $u_{i_t}(\deltat[t+1])=\deltat[t+1](x^+)/v_{i_t,x^+}$ for every  $x^+ \in A^+_{i_t}$. 
An illustrative example is given in \Cref{fig:bestresponsebinary}. In particular, $\deltat[t+1](x^-)/v_{i_t,x^-} \geq u_{i_t}(\deltat[t+1])=\deltat[t+1](x^+)/v_{i_t,x^+}$ for all $x^- \in A^-_{i_t}$ and $x^+ \in A^+_{i_t}$.
\begin{figure}
\begin{center}
\begin{tikzpicture}
    \draw(0,0)--(11,0);
    \draw(0,3pt)--(0,-3pt);
    \node at (0,15pt) {0};

    \draw(3,3pt)--(3,-3pt) node[below] {\textcolor{gray}{$w$}};
    \node at (3,15pt) {3};

    \draw(2,3pt)--(2,-3pt) node[below] {\textcolor{blue!50}{$x$}};
    \node at (2,15pt) {2};
    \draw(5,3pt)--(5,-3pt) node[below] {\textcolor{blue}{$x,y$}};
    \node at (5,15pt) {5};
    \draw[->,color=green!50!black] (2,5pt)--(5,5pt);
    
    \draw(6,3pt)--(6,-3pt) node[below] {\textcolor{blue!50}{$y$}};
    \node at (6,15pt) {6};
    \draw[->,color=red!50!black] (6,5pt)--(5,5pt);
    
    \draw(9,3pt)--(9,-3pt) node[below] {\textcolor{blue!50}{$z$}};
    \node at (9,15pt) {9};
    \draw(7,3pt)--(7,-3pt) node[below] {\textcolor{blue}{$z$}};
    \node at (7,15pt) {7};
    \draw[->,color=red!50!black] (9,5pt)--(7,5pt);
\end{tikzpicture}   
\end{center}
\caption{An instance with four charities (named $w,x,y,z$), $\deltat=(3,2, 6, 9)$, and an agent $i_t$ with $\deltat_{i_t}=(0,2,2,2)$ and Leontief utilities with binary weights $v_{i_t} = (0,1,1,1)$.
 Then, $\delta^\mathit{best}_{i_t}=(0,5,1,0)$, $\deltat[t+1] = (3, 5, 5,7)$,  $c_t=3$, $A^-_{i_t} = \{y,z\}$, $A^+_{i_t} = \{x\}$.}
\label{fig:bestresponsebinary}
\end{figure}

Define $d_i(\delta)$ as the amount of contribution that would be shifted by an agent $i$ if the current distribution (along with its associated decomposition) were $\delta$ and it was her turn to respond. 
Note that we define $d_i(\delta)$
for all agents, not only the one who actually plays her best response; in particular, $d_{i_t}(\deltat)=c_t$ for all $t$.
Note also that $\delta$ is the equilibrium distribution if and only if $d_i(\delta)=0$ for all $i\in N$.

\begin{restatable}{lemma}{lemleonbrdynshift}
\label{lem: leon-br-dyn-shift}
For any sequence $\mathcal{S}$, round $t\ge 0$, and agent $j\in N$,
\begin{align*}
	d_j(\delta^{t}) \le 
	d_{i_{t}}(\delta^{t}) + 
	d_j(\deltat[t+1]).
\end{align*}
\end{restatable}

\begin{proof}
If $d_j(\delta^{t}) \le d_{i_t}(\delta^{t})$, the statement holds trivially.
Hence, assume that $d_j(\delta^{t}) > d_{i_t}(\delta^{t})$.
In particular, $j \ne i_t$.

Let $\Tilde{\delta}_j^{t+1}$ and $\Tilde{\delta}^{t+1}$ be the (hypothetical) individual distribution of agent $j$ and the collective distribution had she been able to implement her best response at round $t$.

Denote the sets of charities that would be affected by agent $j$'s best response at $\deltat$ by $A^-_j\coloneqq\{x^- \in A_j: \Tilde{\delta}_j^{t+1}(x^-)<\delta_j^{t}(x^-)\}$ and $A^+_j\coloneqq\{x^+ \in A_j: \Tilde{\delta}_j^{t+1}(x^+)>\delta_j^{t}(x^+)\}$. Then,
\begin{equation}
\label{cond: leon-br-dyn1}
\begin{split}
\frac{\Tilde{\delta}^{t+1}(x^-)}{v_{j,x^-}} 
\ge \frac{\Tilde{\delta}^{t+1}(x^+)}{v_{j,x^+}} \text{ for all } x^- \in A_j^- \text{ and } x^+ \in A_j^+; \text{ and}
\end{split}
\end{equation}
\begin{align}
\label{cond: leon-br-dyn2}
\Tilde{\delta}^{t+1}_j(x)=0 \text{ for all } x \in A \text{ with } u_j(\Tilde{\delta}^{t+1})<\frac{\Tilde{\delta}^{t+1}(x)}{v_{j,x}} 
\end{align}
hold by definition of best responses.

Now, a lower bound for $d_j(\deltat[t+1])$ is given by the amount shifted from charities in $A^-_j$ under $j$'s best response in round $t+1$. 
Again, denote by $\Tilde{\delta}_j^{t+2}$ and $\Tilde{\delta}^{t+2}$ agent $j$'s best response in round $t+1$ and the corresponding collective distribution;
note that both \eqref{cond: leon-br-dyn1} and \eqref{cond: leon-br-dyn2} hold also with $t+1$ replaced by $t+2$.

Consider first the special case in which
agent $i_t$ did not change her contribution to charities in 
$A_j^- \cup A_j^+$, that is,  $\delta^{t}(x)=\deltat[t+1](x)$ for all $x \in A_j^- \cup A_j^+$. If $d_j(\deltat[t+1])<d_j(\deltat[t])$, then a smaller amount is transferred from charities in $A_j^-$ and to charities in $A_j^+$ in $j$'s best response at $\deltat[t+1]$ than in $j$'s best response at $\deltat$, so by \eqref{cond: leon-br-dyn1}, there exist charities $x^- \in A^-_j$ and $x^+ \in A^+_j$ such that $\Tilde{\delta}^{t+2}(x^-)>\Tilde{\delta}^{t+2}(x^+) 
\ge 
v_{j,x^+} \cdot u_i(\Tilde{\delta}^{t+2})$, and thus $\Tilde{\delta}^{t+2}_j(x^-)>0$.
This contradicts (\ref{cond: leon-br-dyn2}) with $t+2$ instead of $t+1$.
Therefore, $d_j(\deltat[t+1]) \ge d_j(\delta^{t})$ and the claim follows.

Consider now the general case, in which agent $i_t$ may have changed her contribution to some charities in $A_j^- \cup A_j^+$. 
We claim that the total transfer of $i_t$ and then $j$ (i.e., $d_{i_{t}}(\delta^{t}) + d_j(\deltat[t+1])$) cannot be less than the transfer if $j$ were to act alone (i.e., $d_j(\deltat)$).
The reason is similar to the previous paragraph: 
if this total transfer is less than $d_j(\delta^t)$, then there exist charities $x^- \in A^-_j$ and $x^+ \in A^+_j$ such that $\Tilde{\delta}^{t+2}(x^-)>\Tilde{\delta}^{t+2}(x^+) \ge v_{j,x^+}\cdot  u_j(\Tilde{\delta}^{t+2})$ and $\Tilde{\delta}^{t+2}_j(x^-)>0$, which is a contradiction.
Hence, $d_{i_{t}}(\delta^{t}) + d_j(\deltat[t+1]) \ge d_j(\delta^{t})$, as desired.
\end{proof}

Intuitively, the lemma can be seen as a ``triangle inequality'': the left-hand side denotes the direct distance from $\delta^t$ towards $j$'s optimal redistribution, while the right-hand side denotes the distance along an indirect path that first goes to $\delta^{t+1}$ and then proceeds from there towards $j$'s optimal redistribution.

For any agent $j\in N$ and round $t$, we know that $j$ will get the chance to redistribute her contribution in at most $K$ rounds by assumption. Denote this next round by $t \leq t' \le t+K-1$. We have
\begin{align*}
\sum_{\ell=t}^{t'} c_\ell &= \sum_{\ell=t}^{t'} d_{i_\ell}(\delta^\ell)
\\
& \geq 
\sum_{\ell=t}^{t'}\left(d_j(\deltat[\ell])-d_j(\deltat[\ell+1])\right)   && \text{(by \Cref{lem: leon-br-dyn-shift})}
\\
&= d_j(\deltat) - d_j(\deltat[t'+1])  
\\
&= d_j(\deltat) &&
\text{(as $d_j(\deltat[t'+1])=0$ after agent $j$'s best response).}
\end{align*}
Thus, we have an upper bound on the maximum amount any agent would like to shift at any given round $t$.
\begin{corollary}
\label{cor: leon-br-dyn-max}
For all rounds $t$,
$\displaystyle
    \sum_{\ell=t}^{t+K-1} c_\ell 
    \geq 
    \max_{i\in N} d_i(\delta^t).
$
\end{corollary}

\begin{restatable}{lemma}{corleonbrdynshift}
\label{cor: leon-br-dyn-shift-sum}
For any best-response sequence $\mathcal{S}$,
$\lim_{t\to\infty} \sum_{\ell=t}^{t+K-1} c_\ell = 0$.
\end{restatable}

\begin{proof}
Assume for contradiction that there exists $\gamma>0$ such that, for all $T>0$, there exists $T' \ge T$ with $\sum_{\ell=T'}^{T'+K-1} c_\ell \ge \gamma$.
Let $\epsilon:=\gamma/(K(m-1))$.

Recall that $\phi^*$ is the limit of $\Phi(\delta^t)$ as $t\to\infty$.
Choose some $T$ such that $\phi^*-\Phi(\delta^T)<\frac{\epsilon^2}{4 C_N}$ and $T' \ge T$ with $\sum_{\ell=T'}^{T'+K-1} c_\ell \ge \gamma$.
Thus, there exists some $t \in \{T',\ldots, T'+K-1\}$ with $c_{t} \ge \gamma/K$. Consequently, in round $t$, agent $i_t$ transfers at least $\epsilon$ from some charity~$x$ to some other charity $y$.   

The upper bound on $\log(1+x)$ from \Cref{lem:potential-pairwise} can be refined to $\log(1+x)>\frac{x}{1+x}+\frac{x^2}{(2+x)^2}$ for $x>-1$ and $x \neq 0$, so we get

    \begin{align*}
        \Phi(\delta^{t+1})-\Phi(\delta^{T'}) &\ge \Phi(\delta^{t+1})-\Phi(\deltat)
        \\
        &>\deltat(x) \frac{\left(\frac{\epsilon}{\deltat(x)-\epsilon}\right)^2}{\left(2+\frac{\epsilon}{\deltat(x)-\epsilon}\right)^2}
        +\deltat(y) \frac{\left(\frac{-\epsilon}{\deltat(y)+\epsilon}\right)^2}{\left(2+\frac{-\epsilon}{\deltat(y)+\epsilon}\right)^2}
        \\
        &= \deltat(x) \frac{\epsilon^2}{\left(2\deltat(x)-\epsilon \right)^2}
        +  \deltat(y) \frac{\epsilon^2}{\left(2\deltat(y)+\epsilon \right)^2}
        \\
        &>\deltat(x) \frac{\epsilon^2}{\left(2\deltat(x)-\epsilon \right)^2}
        \\
        &> \frac{\epsilon^2}{4\deltat(x)}
        > \frac{\epsilon^2}{4C_N}
        > \phi^*-\Phi(\delta^T)
        \\
        &> \phi^*-\Phi(\delta^{T'}).     
    \end{align*}
It follows that $\Phi(\delta^{t+1}) > \phi^*$. 
However, this is impossible, since $\Phi(\delta^t)$ is increasing with $t$ and converges to $\phi^*$. 
\end{proof}

Combining \Cref{cor: leon-br-dyn-shift-sum} with \Cref{cor: leon-br-dyn-max}, we get that the amount an agent wants to redistribute converges to $0$.

\begin{restatable}{lemma}{corleonbrdynshift2}
\label{cor: leon-br-dyn-shift}
For any best-response sequence $\mathcal{S}$ and agent $j \in N$, it holds that $\lim_{t \to \infty}d_j(\delta_1,\ldots,\delta_n)=0$.
\end{restatable}

We can now complete the proof of \Cref{thm:leon-dynamics}.

\begin{proof}[Proof of \Cref{thm:leon-dynamics}]
The sequence $\mathcal{S}$ is an infinite sequence of vectors of individual distributions, 
so it is a vector in the closed and bounded polytope $\Delta(C_1)\times \cdots \times \Delta(C_n)$.
Hence,  
the Bolzano-Weierstrass theorem states that it has a convergent subsequence $(\delta^{t_k}_1,\ldots,\delta^{t_k}_n)_{k \in \mathbb{N}}$ with limit $(\delta^L_1,\ldots,\delta^L_n)$.
Furthermore, by \Cref{cor: leon-br-dyn-shift}, $\lim_{k \to \infty}d_i(\delta^{t_k}_1,\ldots,\delta^{t_k}_n)=0$ for \emph{any} convergent subsequence, implying $d_i(\lim_{k \to \infty}(\delta^{t_k}_1,\ldots,\delta^{t_k}_n))=0$ for every agent $i \in N$,
and so $d_i(\delta^L_1,\ldots,\delta^L_n) = 0$.
This means that $(\delta^L_1,\ldots,\delta^L_n)$ is a Nash equilibrium.

Therefore, for every convergent subsequence, the sum of individual distributions in the limit is the same and equals the unique equilibrium distribution $\delta^*$.
It follows that the sequence of sums (i.e., collective distributions) in $\mathcal{S}$ itself must converge to the same limit $\delta^*$---otherwise, by the Bolzano-Weierstrass theorem, we would have a subsequence that converges to a different sum.
\end{proof}

\begin{remark}
    \Cref{thm:leon-dynamics} only shows that the \emph{collective} distributions generated by $\mathcal{S}$ converge to the equilibrium distribution.
    We do not know whether the vector of \emph{individual} distributions converges to a Nash equilibrium for arbitrary $\mathcal{S}$.
    In theory, we might have different convergent subsequences that converge to different Nash equilibria (which all yield the collective equilibrium distribution).
\end{remark}

\begin{remark}
As explained in \Cref{ftn:better-response}, \Cref{lem:potential} does not hold for better-response sequences. However, the statement remains true when restricting ourselves to transfers from one charity to another as in \Cref{lem:potential-pairwise}. For Cobb-Douglas utilities (or lexicographic Leontief utilities), such transfers correspond to better responses.
\end{remark}

\subsection{Proof of \Cref{thm:contspenddynamics}}

\contspenddynamics*

\begin{proof}
   For every $t$, note that $\delta^\mathit{best}_{t}$ is the same distribution as the best response of agent~$i_t$ under the redistribution dynamics of \Cref{sub:redistribution-dynamics} with round-robin sequence $\mathcal{S}$. 
   Thus, \Cref{thm:leon-dynamics} implies that the sum of the last $n$ individual distributions (one per agent) converges to the equilibrium distribution, i.e., $\lim_{t \to \infty}\sum_{k=t-n+1}^{t}\delta^\mathit{best}_{k}=\delta^*$.
    Consequently, for $t$ being a multiple of $n$, the sum
    \begin{align*}
    \sum_{i \in N}\frac{1}{\lfloor (t+n-i)/n \rfloor} \delta^t_i=\sum_{i \in N}\frac{n}{t}\delta^t_i=
    \frac{n}{t}\sum_{\ell=1}^{t/n}\sum_{k=(\ell-1) n+1}^{\ell n}\delta^\mathit{best}_{k}
    \end{align*}
    converges to $\delta^*$ as $t\to\infty$.
    
    For arbitrarily large $t$ not being a multiple of $n$, donations from rounds $\lfloor t/n \rfloor n, \dots,t-1$ 
    only have a vanishing impact on $\sum_{i \in N}\frac{1}{\lfloor (t+n-i)/n \rfloor} \delta^t_i$
for $t\to\infty$, as the size of each donation is constant and the denominator increases with $t$.
    Hence, convergence to $\delta^*$ holds for the whole sequence. 
\end{proof}

\section{Proofs omitted from Section \ref{sec:binary}}
\label{app: proofs}

\edrisprojegal*

\begin{proof}
By uniqueness of the equilibrium distribution (\Cref{thm:unique-eq-distribution}), it is sufficient to show that it coincides with the charity egalitarian distribution.
Let $\delta^\mathit{CHEG}$ be the decomposable charity egalitarian distribution, with decomposition $\delta^\mathit{CHEG} = \sum_{i\in N} \delta^\mathit{CHEG}_i$. 
Suppose for contradiction that $\delta^\mathit{CHEG}$ is not the equilibrium distribution. 
By \Cref{lem:eq-iff-critical}, there is an agent $i\in N$ who contributes to a non-critical charity $x \in A_i$, that is, $\delta^\mathit{CHEG}_i(x)>0$ and $\delta^\mathit{CHEG}(x) > u_i(\delta^\mathit{CHEG})$. 
Let $y\in A_i$ be a critical charity of agent~$i$, that is,
$\delta^\mathit{CHEG}(y) = u_i(\delta^\mathit{CHEG})$.

If agent $i$ now moves $\nicefrac{1}{2}(\delta^\mathit{CHEG}(x) - \delta^\mathit{CHEG}(y))$ from $x$ to $y$, the resulting distribution is still decomposable, as both $x$ and $y$ are in $A_i$. It is leximin-higher than $\delta^\mathit{CHEG}$, contradicting the leximin-maximality of $\delta^\mathit{CHEG}$.
\end{proof}

\edrisceg*

\begin{proof}
By uniqueness of the equilibrium distribution (\Cref{thm:unique-eq-distribution}), it is sufficient to show that it coincides with the conditional egalitarian distribution.
Let $\delta^{\ceg}$ be the conditional egalitarian distribution with decomposition $\delta^{\ceg} = \sum_{i\in N} \delta^{\ceg}_i$. 
Suppose for contradiction that $\delta^{\ceg}$ is not the equilibrium distribution. 
Then, some agent $i\in N$ contributes to a non-critical charity $x\in A_i$, that is, $\delta^{\ceg}_i(x)>0$ and $\delta^{\ceg}(x)>u_i(\delta^{\ceg})$.

Let $D \coloneqq \min\big(\delta^{\ceg}_i(x),~ \delta^{\ceg}(x)-u_i(\delta^{\ceg})\big)$; our assumptions imply that $D>0$.
Construct a new distribution $\delta'$ from $\delta^{\ceg}$ by changing only $\delta^{\ceg}_i$: remove 
$D$ from charity $x$, and add 
$D / |A_i|$ to every charity in $A_i$ (including $x$).
The utility of $i$ increases by 
$D /|A_i|$, since:
\begin{itemize}
\item $\delta'(x) = \delta^{\ceg}(x)-D+D/|A_i| 
\geq u_i(\delta^{\ceg}) + D/|A_i|$ by definition of $D$;
\item $\delta'(y) = \delta^{\ceg}(y)+D/|A_i| \geq u_i(\delta^{\ceg})+D/|A_i|$ for all $y\in A_i\setminus x$, by \eqref{eq:min-based} with equality for $y \in T_{\delta^{\ceg},i}$;
\item So $u_i(\delta') = \min(\delta'(x), \min_{y\in A_i\setminus x}\delta'(y)) = u_i(\delta^{\ceg}) + D/|A_i| > u_i(\delta^{\ceg})$.
\end{itemize}
Moreover, if the utility of some agent $j$ decreases---that is, $u_j(\delta')<u_j(\delta^{\ceg})$---then this must be because of the decrease in the distribution to $x$, so $x$ must be a critical charity for agent $j$ in $\delta'$, i.e., 
$u_j(\delta')=\delta'(x) \geq u_i(\delta') > u_i(\delta^{\ceg})$.

Thus, moving from $\delta^{\ceg}$ to $\delta'$, the number of agents with utility larger than $u_i(\delta^{\ceg})$ strictly increases, and the utility of each agent with utility at most $u_i(\delta^{\ceg})$ in $\delta^{\ceg}$ does not decrease. 
Therefore, $\mathbf{u}(\delta') \succ_\mathit{lex} \mathbf{u}(\delta^{\ceg})$.
Since $\delta'$ is decomposable, this contradicts the optimality of $\delta^{\ceg}$.
\end{proof}

\fmaximizingisefficient*

\begin{proof}[Proof sketch]
Suppose for contradiction that $\delta$ is not efficient.
By Lemma \ref{lem:efficient-iff-critical}, there is a charity $x \in \supp(\delta)$ which is not critical for any agent. Then, one agent who contributes to $x$ would be able to shift a small amount uniformly to the set of her critical charities such that $x$ is still not critical for any agent. The resulting distribution is still decomposable, and Pareto dominates $\delta$, contradicting the maximality of $\delta$ in \gwelfare.

Uniqueness is proved similarly to \Cref{lem:fmaximizing-is-unique}, using the fact that the set of decomposable distributions is convex, i.e., mixing decomposable distributions results in another decomposable distribution. 
\end{proof}

\subsection{Proof of \Cref{thm: cd-is-fmaximal}}\label{app: equivalenceproof}

\maxfwelfareequivalence*
The proof requires some additional definitions and lemmas and proceeds as follows. First, we show that it is sufficient to prove the statement for \emph{reduced} profiles (\Cref{def: reduced-profiles} and \Cref{lem: reduced-profiles}), which are profiles in which each agent approves only charities that receive the same amount in the equilibrium distribution.
Then we prove that, in any reduced profile, the equilibrium distribution $\delta^*$ maximizes \gwelfare, not only in the set of decomposable distributions but even in a larger set of \emph{weakly decomposable} distributions (\Cref{def:weakly-decomposable}).
To do this, we prove that, for any weakly decomposable distribution $\delta \neq \delta^*$, there exists a modification $\delta'$, which is  weakly decomposable
but has a higher \gwelfare{} than $\delta$.

Recall that $[z] \coloneqq \{1,2,\dots,z\}$ for each positive integer $z$.

\begin{definition}
\label{def: charity-partition}
Given any distribution $\delta$,
define $\mathcal{P}(\delta)$ as a partition of the charities into subsets allocated the same amount.
That is, $\mathcal{P}(\delta)\coloneqq(X_1,\ldots,X_p)$ for some integer $p\geq 1$,
where $\cup_{k=1}^p X_k = A$,
and for each $k\in [p]$,
all charities in $X_k$ receive the same amount,
$\delta(x)=w_k$ for all $x\in X_k$,
and the amounts are ordered such that $0\leq w_1< \dots < w_k$.
\end{definition}
Note that $w_1=0$ if and only if there exist charities that receive no funding.
\begin{lemma}
\label{lem: charity-partition}
Let $\delta^*$ be the equilibrium distribution, and $(X^*_1,\ldots,X^*_p) = \mathcal{P}(\delta^*)$ be its charity partition.
For each $k\geq 1$, let $N^*_k$ be the set of agents who approve one or more charities of $X^*_k$, but do not approve any charity of $\cup_{\ell<k}X^*_{\ell}$.
Then in equilibrium, the agents of $N^*_k$ contribute only to charities of $X^*_k$, and:
\begin{align*}
\delta^*(X^*_k) &= C_{N^*_k} \text{, and}
\\
w^*_k &= C_{N^*_k}/|X^*_k| = \delta^*(X^*_k)/|X^*_k|.
\end{align*}
\end{lemma}

\begin{proof}
The utility of all agents in $N^*_k$ is $w^*_k$, so the set of their critical charities is contained in $X^*_k$. 
In equilibrium, they contribute only to charities in $X^*_k$ by \Cref{lem:eq-iff-critical}.

All charities in $X^*_k$ receive the same amount, so this amount must be $C_{N^*_k}/|X^*_k|$.
\end{proof}
Note that if there are charities 
not approved by any agent (or approved only by agents who contribute $0$), then all these charities will be in $X^*_1$, and we will have $w^*_1=C_{N^*_1}=0$.

\begin{definition}\label{def:weakly-decomposable}
A distribution $\delta$ is called \emph{weakly decomposable} if it has a decomposition in which each agent $i$ only contributes to charities $x$ with $\delta^*(x)\geq u_i(\delta^*)$, where $\delta^*$ denotes the equilibrium distribution.
\end{definition}
With binary weights, 
$x \in A_i$ implies 
$\delta^*(x)\geq u_i(\delta^*)$, so every decomposable distribution is weakly decomposable.
Therefore, it is sufficient to prove that $\delta^*$ maximizes \gwelfare among all weakly decomposable distributions.

The set of weakly decomposable distributions is again convex and can be characterized as follows.

\begin{lemma}\label{lem: weak-dec-part}
A distribution $\delta$ is weakly decomposable if and only if, 
for every $\ell \in [p]$,
\begin{align}
\label{eq:weak-dec-part}
\delta\left(\cup^p_{k=\ell}X^*_k\right) \geq \delta^*\left(\cup^p_{k=\ell}X^*_k\right).
\end{align}
\end{lemma}
\begin{proof}
A distribution $\delta$ is weakly decomposable if and only if there exists a decomposition of~$\delta$ where for every $\ell \in [p]$, agents of $N^*_\ell$ only contribute to charities of $\cup^p_{k=\ell}X^*_k$. 
This holds if and only if 
$\delta\left(\cup^p_{k=\ell}X^*_k\right) \geq \sum_{k=\ell}^{p} C_{N^*_\ell}$ for every $\ell \in [p]$. 
By \Cref{lem: charity-partition},
this is equivalent to the condition $\delta\left(\cup^p_{k=\ell}X^*_k\right) \geq \delta^*\left(\cup^p_{k=\ell}X^*_k\right)$ for every $\ell \in [p]$. 
\end{proof}

To simplify the proof of \Cref{thm: cd-is-fmaximal}, we introduce the following class of profiles.

\begin{definition}\label{def: reduced-profiles}
A profile is called \emph{reduced} if, in its equilibrium distribution $\delta^*$, for every agent $i$, there exists a $k \in [p]$ such that $A_i \subseteq X_k^*$, that is, all charities approved by an agent belong to the same class in the partition induced by $\delta^*$.
\end{definition}
Note that, in a reduced profile, 
all charities approved by agent $i$ receive the same amount $u_i(\delta^*)$ under the equilibrium distribution,
and therefore 
are all critical for $i$, that is, $T_{\delta^*,i} = A_i$ for all $i\in N$.
\begin{lemma}
\label{lem: reduced-profiles}
If \Cref{thm: cd-is-fmaximal} is true for reduced profiles, then it is true for all profiles.
\end{lemma}
\begin{proof}
Let $P$ be any profile, and $\delta^*$ its equilibrium distribution. Let $P'$ be its reduced profile where, compared to $P$, every agent $i$ has removed her approval from every charity $x$ with $\delta^*(x)>u_i(\delta^*)$.
Then, $\delta^*$ is the equilibrium distribution for $P'$ as well (by the same decomposition). By assumption, \Cref{thm: cd-is-fmaximal} is true for $P'$, so $\delta^*$  maximizes \gwelfare among all distributions that are weakly decomposable with respect to $P'$.
Since the equilibrium distribution is the same in $P$ and $P'$, the set of weakly decomposable distributions is the same too.

The profile~$P$ differs from $P'$ by having additional approvals, which could only decrease the maximal possible \gwelfare.
But $\delta^*$ yields the same welfare in $P$ and $P'$.
Therefore, $\delta^*$ necessarily maximizes \gwelfare among all distributions that are weakly decomposable with respect to $P$, too. 
\end{proof}

\begin{proof}[Proof of \Cref{thm: cd-is-fmaximal}]
Based on \Cref{lem: reduced-profiles}, we assume without loss of generality that we are given a reduced profile.
Let $X^*_1,\dots, X^*_p$ and $N^*_1,\dots,N^*_p$ be the partitioning of charities and agents induced by the equilibrium distribution $\delta^*$,
and $w^*_1 < \dots < w^*_p$ the corresponding allocations.
By \Cref{lem: charity-partition},
each charity in $X^*_k$ receives $w^*_k = \delta^*(X^*_k)/|X^*_k|$, and every  agent $i\in N^*_k$ has utility $w^*_k$. 
Since the profile is reduced, $T_{\delta^*,i} = A_i\subseteq X^*_k$ for all $i\in N^*_k$.

Let $\delta$ be any weakly decomposable distribution different than $\delta^*$.
We prove that $\delta$ does not maximize \gwelfare among weakly decomposable distributions by showing a modification $\delta'$ of $\delta$, which is weakly decomposable
but has a higher \gwelfare{} than $\delta$.

Since $\delta\neq \delta^*$ and both distributions sum up to $C_N$,
there must be charities $x^-,x^+\in A$  with $\delta(x^-)<\delta^*(x^-)$ and $\delta(x^+)>\delta^*(x^+)$. 
Consequently, one of the following two cases has to apply:
\begin{itemize}
\item If $\delta(X^*_k)=\delta^*(X^*_k)$ for all $k \in [p]$, let $X^*_r=X^*_s$ $(r=s)$ be a class that contains a charity $x^-$ with $\delta(x^-)<\delta^*(x^-)$. 
\item Otherwise, 
let $r$ be the largest index in $[p]$ for which $\delta(X^*_r)\neq \delta^*(X^*_r)$.
Weak decomposability of $\delta$
and \Cref{lem: weak-dec-part} imply that
$\delta(X^*_r)>\delta^*(X^*_r)$.
As $\delta(X^*_k)=\delta^*(X^*_k)$ for all $k>r$, there must be an $s\leq r$ 
such that there exists a charity $x^-$ in $X^*_s$ with $\delta(x^-)<\delta^*(x^-)$;
choose $s \leq r$ to be the largest index with this property.
\end{itemize}
In both cases, we define $X^- \subseteq X^*_s$ as the set of all charities $x$ in $X^*_s$ with $\delta(x)<\delta^*(x)$, and $X^+ \subseteq X^*_r$ as the set of all charities $x$ in $X^*_r$ with $\delta(x)>\delta^*(x)$; both sets must be non-empty by construction.
The case $r>s$ is depicted in \Cref{fig:fwelfare}.

Starting from $\delta$,
transfer a sufficiently small amount $\epsilon$ uniformly from $X^+$ to $X^-$; call the resulting distribution $\delta'$.
We choose $\epsilon$ small enough such that it does not change the order relations between charities inside and outside $X^+$ and $X^-$, that is, for all $x^- \in X^-$ and $x^+ \in X^+$: 
$\delta'(x^+)>\delta'(x)$
for all 
$x \in A$ with $\delta(x^+)>\delta(x)$, 
and analogously, $\delta'(x^-)<\delta'(x)$ for all 
$x \in A$ with $\delta(x^-)<\delta(x)$. 
In particular, since $\delta(x^+)>\delta^*(x^+)\ge \delta^*(x^-) > \delta(x^-)$, we have  $\delta'(x^+) > \delta'(x^-)$.

We claim that $\delta'$ is weakly decomposable. 
By \Cref{lem: weak-dec-part}, it suffices to show that \eqref{eq:weak-dec-part} 
holds for 
$\delta'$, that is,
$\delta'\left(\cup^p_{k=\ell}X^*_k\right) \geq \delta^*\left(\cup^p_{k=\ell}X^*_k\right)$
for every $\ell \in [p]$.
Note that
$\delta'\left(\cup^p_{k=\ell}X^*_k\right)=\delta\left(\cup^p_{k=\ell}X^*_k\right)$ for all $\ell \leq s$ and all $\ell \geq r+1$, so for these indices, 
\eqref{eq:weak-dec-part} for $\delta'$ follows from the weak-decomposability of~$\delta$.
It therefore remains to prove 
\eqref{eq:weak-dec-part} for $\ell\in \{s+1,\ldots, r\}$.
This set is non-empty only when $s<r$, which is possible only in the second case above. 

Our choices of $r$ and $s$ ensure that $\delta\left(\cup^p_{k=r}X^*_k\right) > \delta^*\left(\cup^p_{k=r}X^*_k\right)$ and $\delta\left(\cup^s_{k=1}X^*_k\right) < \delta^*\left(\cup^s_{k=1}X^*_k\right)$. 
For $\epsilon$ sufficiently small, the same inequalities hold between $\delta'$ and $\delta^*$. 
Moreover, for $s<\ell\le r$,
\begin{align*}
     \delta'\left(\cup^p_{k=\ell}X^*_k\right)= \delta'\left(\cup^p_{k=r}X^*_k\right)+
     \delta'\left(\cup^{r-1}_{k=\ell}X^*_k\right)>
\delta^*\left(\cup^p_{k=r}X^*_k\right)+
\delta\left(\cup^{r-1}_{k=\ell}X^*_k\right) \geq
\delta^*\left(\cup^p_{k=\ell}X^*_k\right)
\end{align*}
where the first inequality holds because $\delta'\left(\cup^p_{k=r}X^*_k\right) > \delta^*\left(\cup^p_{k=r}X^*_k\right)$ and $\delta'(X^*_k)=\delta(X^*_k)$ for all $k \not\in\{r,s\}$, and the second inequality holds because, for each $k\in\{s+1,\dots,r-1\}$, all charities $x$ in $X^*_k$ satisfy $\delta(x) \ge \delta^*(x)$ by definition of $s$.
Therefore, by \Cref{lem: weak-dec-part}, $\delta'$ is still weakly decomposable.

\begin{figure}
\begin{center}
\begin{tikzpicture}[scale=0.8]
\draw[->] (-1,0) -- (8.3,0);
\draw[->] (0,-0.5) -- (0,8);
\node at (8,-0.4) {\large $\delta^*$};
\node at (-0.3,7.7) {\large $\delta$};
\node at (1,-0.5) {\large $X_1^*$};
\node at (3,-0.5) {\large $X_s^*$};
\node at (5,-0.5) {\large $X_r^*$};
\node at (7,-0.5) {\large $X_p^*$};
\draw [blue,fill=blue!10] (0.5,0.5) rectangle (1.5,7.5);
\draw [blue,fill=blue!10] (2.5,0.5) rectangle (3.5,7.5);
\draw [blue,fill=blue!10] (4.5,0.5) rectangle (5.5,7.5);
\draw [blue,fill=blue!10] (6.5,0.5) rectangle (7.5,7.5);
\draw[thick] (2.5,2) -- (3.5,2);
\node at (3,1.2) {$X^-$};
\draw[fill=black] (2.9,2) circle  [radius=0.08];
\node at (3,2.4) {$x^-_{\max}$};
\draw[thick] (4.5,5) -- (5.5,5);
\node at (5,6.2) {$X^+$};
\draw[fill=black] (4.9,5) circle  [radius=0.08];
\node at (5,4.6) {$x^+_{\min}$};
\end{tikzpicture}
\end{center}
\caption{Charity sets in the proof of \Cref{thm: cd-is-fmaximal}, for the case $r>s$.
The horizontal position of a charity denotes its allocation in $\delta^*$;
the vertical position denotes its allocation in $\delta$.
}
\label{fig:fwelfare}
\end{figure}

We now analyze the effect of this redistribution on the agents' utilities. For this, we prove an auxiliary claim on critical charities of agents under $\delta$.
Define $x^+_{\min}\in\argmin_{x^+ \in X^+}\delta(x^+)$ as a charity from $X^+$ with minimal allocation in $\delta$ and $x^-_{\max}\in\argmax_{x^- \in X^-}\delta(x^-)$ as a charity from $X^-$ with maximal contribution in $\delta$.

\paragraph{Claim.}
For every agent $i\in N$, either $T_{\delta,i}\cap X^- = \emptyset$ or $T_{\delta,i}\subseteq X^-$. 
Similarly, 
either $T_{\delta,i}\cap X^+ = \emptyset$ or $T_{\delta,i}\subseteq X^+$. 

\paragraph{Proof of claim.}
We prove the claim for $X^-$; the proof for $X^+$ is analogous.
By definition of critical charities, $T_{\delta,i}\subseteq A_i$. 
Since the profile is reduced, $A_i$ is contained in a single partition class.
If this partition class is not the one that contains $X^-$, namely $X^*_s$, then $T_{\delta,i}\cap X^- = \emptyset$. Otherwise, 
$T_{\delta,i}\subseteq X_s^*$.
Now, if $u_i(\delta) > \delta(x^-_{\max})$,
then $\delta(x)  > \delta(x^-_{\max})$ 
for every $x\in T_{\delta,i}$,
so $T_{\delta,i}\cap X^- = \emptyset$;
and if $u_i(\delta) \leq \delta(x^-_{\max})$,
then $\delta(x) \leq \delta(x^-_{\max})$ for every $x\in T_{\delta,i}$,
so $T_{\delta,i}\subseteq X^-$.

\paragraph{Back to proof of theorem.} 
Denote by ``losers'' the agents who lose utility from the redistribution. The claim implies that all the losers have $T_{\delta,i} \subseteq X^+$; each of them loses  $\epsilon/|X^+|$.
Moreover, all losers have $A_i\subseteq X^+$:
this is because $A_i\subseteq X_r^*$ (since the profile is reduced),
and $\delta(x_A)\geq \delta(x_T) \geq \delta(x^+_{\min})$ for all $x_A\in A_i$ and $x_T\in T_{\delta,i}$.
Therefore, under the equilibrium distribution, all losers give all their contributions to charities in $X^+$. 
This implies that the contributions of all losers sum up to at most  $\delta^*(X^+) = w^*_r\cdot |X^+|$.
Then, for every loser $i$,
\begin{align}
\label{ineq: losers}
g\left(u_i(\delta)\right)-g\left(u_i(\delta')\right)\leq 
g\left(\delta(x^+_{\min}) \right)- g\left(\delta(x^+_{\min})-\frac{\epsilon}{|X^+|} \right)
\end{align}
by concavity of $g$ (which follows from the assumption that $x g'(x)$ is non-increasing).

Denote by ``gainers'' the agents who gain utility from the redistribution.
The claim implies that 
every agent with $T_{\delta,i} \cap X^-\neq\emptyset$ is a gainer;
each of them gains $\epsilon/|X^-|$.
Moreover, every agent with $A_i \cap X^-\neq\emptyset$ is a gainer:
this is because $A_i \cap X^-\neq\emptyset$ implies 
$\delta(x_A) \leq \delta(x^-_{\max})$ for at least one charity $x_A\in A_i$,
and $\delta(x_T)\leq \delta(x_A)$ for all charities $x_T\in T_{\delta,i}$.
Therefore, under the equilibrium distribution, every agent who contributes a positive amount to at least one charity in $X^-$ must be a gainer. 
So the contributions of all gainers must sum up to at least $\delta^*(X^-) = w^*_s\cdot |X^-|$.
Then, for every gainer $i$,
\begin{align}
\label{ineq: winners}
g\left(u_i(\delta')\right)-g\left(u_i(\delta)\right)\geq 
g \left(\delta(x^-_{\max})+\frac{\epsilon}{|X^-|} \right)-g \left(\delta(x^-_{\max}) \right)
\end{align}
by concavity of $g$.

Therefore, by (\ref{ineq: losers}) and (\ref{ineq: winners}), the increase in \gwelfare from $\delta$ to $\delta'$ is at least
\begin{align}
\label{ineq: eps-shift}
&
w^*_s\cdot|X^-| \cdot\left[ g \left(\delta(x^-_{\max})+\frac{\epsilon}{|X^-|} \right)-g \left(\delta(x^-_{\max}) \right) \right]
\\ \notag
&-
w^*_r \cdot|X^+| \cdot\left[g\left(\delta(x^+_{\min}) \right)- g\left(\delta(x^+_{\min})-\frac{\epsilon}{|X^+|} \right)\right]. 
\end{align}
Since $g$ is strictly concave,
\begin{align*}
    g \left(\delta(x^-_{\max})+\frac{\epsilon}{|X^-|} \right)-g (\delta(x^-_{\max})) &>\frac{\epsilon}{|X^-|}\cdot g' \left(\delta(x^-_{\max})+\frac{\epsilon}{|X^-|} \right);
    \\
    g(\delta(x^+_{\min}))-
    g\left(\delta(x^+_{\min})-\frac{\epsilon}{|X^+|}\right) &< 
    \frac{\epsilon}{|X^+|}\cdot g'\left(\delta(x^+_{\min})-\frac{\epsilon}{|X^+|} \right).
\end{align*}
Plugging this into (\ref{ineq: eps-shift}), 
we get that 
the increase in \gwelfare  is larger than
\begin{align*}
w^*_s\cdot|X^-|\cdot \frac{\epsilon}{|X^-|}\cdot g' \left(\delta(x^-_{\max})+\frac{\epsilon}{|X^-|} \right)
-
w^*_r\cdot|X^+|\cdot\frac{\epsilon}{|X^+|}\cdot g'\left(\delta(x^+_{\min})-\frac{\epsilon}{|X^+|} \right).
\end{align*}

By our choice of $\epsilon$,
we have 
$w^*_r=\delta^*(x^+_{\min}) < \delta(x^+_{\min})$,
so
$w^*_r<\delta(x^+_{\min})-\epsilon/|X^+|$
for sufficiently small $\epsilon$.
Similarly, 
 $w^*_s=\delta^*(x^-_{\max}) > \delta(x^-_{\max})+\epsilon/|X^-|$ for sufficiently small~$\epsilon$.
Therefore, the increase in \gwelfare is larger than
\begin{align}
\begin{split}
\label{eq:fwelfare-increase}
&\epsilon\cdot\left(\delta(x^-_{\max})+\frac{\epsilon}{|X^-|}\right)\cdot g' \left(\delta(x^-_{\max})+\frac{\epsilon}{|X^-|} \right)
\\
&-\epsilon\cdot\left(\delta(x^+_{\min})-\frac{\epsilon}{|X^+|}\right)\cdot g'\left(\delta(x^+_{\min})-\frac{\epsilon}{|X^+|} \right).
\end{split}
\end{align}
By our choice of $\epsilon$, $\delta(x^-_{\max})+\epsilon/|X^-| < \delta(x^+_{\min})-\epsilon/|X^+|$.
By the assumption on $g$, $xg'(x)$ is non-increasing in $x$.
Therefore, the expression in \eqref{eq:fwelfare-increase} is at least $0$, 
so the increase in \gwelfare from $\delta$ to~$\delta'$ is larger than $0$. 
This means that $\delta$ does not maximize \gwelfare.

Since $\delta$ was any weakly decomposable distribution different than $\delta^*$, we conclude that $\delta^*$ maximizes \gwelfare subject to weak decomposability in any reduced profile.
By \Cref{lem: reduced-profiles}, 
the same is true in any profile.
\end{proof}

\subsection{Proof of \Cref{prop: p>0}}\label{app: p>0}

\pwelfarelemma*
\begin{proof}
For a fixed $p>0$, consider a profile consisting of two agents with binary weights and approval sets $\{a\}$ and $\{a,b\}$, and respective contributions $C_1=\max\left({(2^{p-1}\cdot p)}^{\nicefrac{-1}{p}},2\right)$ and $C_2=1$. 
Since $C_1 \ge 2$, the equilibrium distribution is $(C_1,1)$.
We claim that the decomposable distribution $(C_1+1,0)$ yields a higher \gwelfare, that is,
\begin{align*}
&    C_1\cdot g(C_1+1) + 1\cdot g(0) >C_1 \cdot g(C_1) +1\cdot g(1)
    \\
\Leftrightarrow \quad &    
 C_1 \cdot \left(g(C_1+1) - g(C_1)\right) > 1.
\end{align*}
For every $p \ge 1$, $g$ is convex, so 
\begin{align*}
&g(C_1+1) - g(C_1) \geq g'(C_1)\cdot 1 = p\cdot C_1^{p-1}
\\
\Rightarrow \quad &
C_1\cdot \left(g(C_1+1) - g(C_1)\right) \geq p\cdot C_1^{p} \geq p\cdot 2^p \geq 2 > 1.
\end{align*}
For every $0<p<1$, $g$ is strictly concave, so 
\begin{align*}
g(C_1+1) - g(C_1) > g'(C_1+1)\cdot 1 &= p\cdot (C_1+1)^{p-1}
\\
& > p \cdot (2 C_1)^{p-1}  && \text{(since $p-1<0$ and $C_1 > 1$)}
\end{align*}
and therefore
\begin{align*}
C_1\cdot \left(g(C_1+1) - g(C_1)\right) &> 2^{p-1}\cdot p\cdot C_1^{p} 
\\
&\geq 
 2^{p-1} \cdot p \cdot \left(\left(\frac{1}{2^{p-1}p}\right)^{\frac{1}{p}}\right)^p = 1.
\end{align*}
In both cases, the equilibrium distribution does not maximize \gwelfare. 
\end{proof}

\end{document}